\documentclass[12pt]{article}

\usepackage{amsmath}
\usepackage{amssymb}
\usepackage{amsthm}
\usepackage{epsfig}

\newcommand{\prt}{\partial}

\def\Pr{{\mathbb P}}

\def\B{{\mathcal B}}
\def\V{{\mathcal V}}
\def\A{{\mathcal A}}
\def\D{{\mathcal D}}

\def\T{{\Phi}}

\def\R{\ensuremath{\mathbb{R}}}

\def\Pe{\mathcal{P}}

\def\be{\begin{equation}}
\def\ee{\end{equation}}

\def\ni{\noindent}

\newtheorem{theorem}{Theorem}

\newtheorem{lemma}{Lemma}
\newtheorem{corollary}[theorem]{Corollary}

\theoremstyle{definition}

\begin{document}
\bibliographystyle{plain}

\title{{\Large\bf  Various equations for gap probabilities of coupled Gaussian matrices}}
\author{Igor Rumanov\footnote{e-mail: igorrumanov@math.ucdavis.edu} \\
{\small Department of Mathematics, UC Davis, 1 Shields Avenue, CA 95616}}

\maketitle

\bigskip

\begin{abstract}
\par Versions of Tracy-Widom (TW) and Adler-Shiota-van Moerbeke (ASvM) approaches are applied to derive various partial differential equations (PDE) satisfied by joint eigenvalue spacing probabilities of two coupled Gaussian Hermitian matrices (coupled GUE). All the lowest (third) order PDE satisfied by the probabilities for the largest eigenvalues of two coupled GUE are found, and the results of both approaches are compared. 
\par The TW approach allows to derive all PDE at once, while in the ASvM one starting with different bilinear identities leads to different subsets of the full set of equations. 
\par An interesting result is that the joint probability for the largest eigenvalues of coupled Gaussian matrices satisfies a number of different PDE, and the previously known Adler-van Moerbeke equation (AvM)~\cite{AvM1} is only one of them. Some of the new equations look like ``coupled Painlev\'e IV" and have usual Painlev\'e IV equation as one-matrix limit, i.e.~when the spectral endpoint of one of the matrices goes to infinity. This is in contrast to the AvM equation, which becomes trivial in this limit. Moreover, the new PDE, which stem from the matrix kernel approach of~\cite{TW2}, do not contain derivatives w.r.t. the strength of coupling, unlike the AvM equation. In other words, they contain fewer independent variables and in this sense are simpler.
\end{abstract}

\newpage 

\section{Introduction}

The problem of studying the joint eigenvalue probability distribution for coupled Gaussian matrices arises e.g.~in studying {\it Dyson process}. A {\it Dyson process} is any process on ensembles of matrices in which the entries undergo diffusion. Consider ensemble of $n\times n$ Hermitian Gaussian matrices. Then $n$ curves describe changing eigenvalues with time. The probability that for each $k$ the $k$-th eigenvalue lies outside of a given set $J_k$ at time $t_k$ is given by the Fredholm determinant of a certain matrix kernel, called the {\it extended Hermite kernel}~\cite{TW2}, which is a particular case of a general matrix kernel first obtained by Eynard and Mehta~\cite{EyMe}. Scaling this {\it Gaussian process} at the edge leads to the {\it Airy process}, introduced by Pr\"ahofer and Spohn~\cite{PrSp} as the limiting stationary process for a polynuclear growth model. It is conjectured to be the limiting process for a wide class of random growth models, in particular, those belonging to the KPZ universality class~\cite{KPZ}. Recently steps were made by Tracy and Widom~\cite{TW-ASEP3} toward establishing it for the Asymmetric Simple Exclusion Process (ASEP), see e.g.~\cite{Lig}, when they proved the long-standing conjecture that ASEP belongs to the KPZ universality class in a scaling limit of long time and large number of particles and its single point limiting distribution is the celebrated $F_2$~\cite{TW-Airy}. In earlier works~\cite{TW-AiryPr, TW2} Tracy and Widom (the approach in these works will be called TW further on) found systems of ordinary differential equations (ODE) for Airy, Gaussian and some other Dyson processes. Also, a single PDE for Gaussian process was known for a while from the work of Adler and van Moerbeke~\cite{AvM1} (AvM), as well as PDE for Airy and Sine processes found from it in the corresponding scaling limits~\cite{AvM3}. The last authors used two-dimensional Toda lattice (2-Toda) bilinear identity and Virasoro constraints to derive their equations.
\par Here we compare the above two approaches to the joint distribution of the spectrum for coupled matrices. Even though we use a modified version of AvM approach, analogous to what we used for one-matrix ensembles in~\cite{IR1, UniUE}, which is closer to TW approach and is more flexible, we still find that system of equations derived from matrix kernel contains more information and allows in fact to obtain {\it all possible} PDE which such probabilities satisfy. This is quite similar to what we recently found for the general one-matrix case with several spectral endpoints~\cite{Multipt1matr}.
\par As Dyson first observed long ago~\cite{DyBM}, for Hermitian matrices with matrix elements independently executing Brownian motion subject to a restoring harmonic force, the equilibrium measure is the GUE measure of random matrix theory. With initial conditions at time $t_1$ distributed according to the GUE measure, the probability that at times $t_k$ ($k = 2, \dots, m$, we assume that $t_1 < t_2 < \dots < t_m$) the matrix $M(t_k)$ is infinitesimally close to a value $M_k$ is proportional to

\be
\Pr \sim \exp(-Tr M_1^2)\prod_{k=2}^m \exp\left(-\frac{Tr (M_k - c_{k-1}M_{k-1})^2}{1-c_{k-1}^2}\right)dM_1\cdots dM_m,  \label{eq:00.1}
\ee

\noindent where $c_k = e^{t_k-t_{k+1}}$. Alternatively, (\ref{eq:00.1}) can be interpreted as the equilibrium measure for a chain of $m$ coupled $n\times n$ Hermitian matrices $M_k$. The distribution of eigenvalues for this measure is expressible as the Fredholm determinant of an $m\times m$ matrix kernel~\cite{EyMe} related to the kernel associated with the equilibrium random matrix ensemble -- GUE in our case. It is derived by diagonalizing each $M_k$ and then using the celebrated Harish-Chandra/Itzykson-Zuber (HCIZ) formula~\cite{HC, IZ}:



\be
\int_{U(n)}e^{\text{Tr}(XUYU^{-1})}dU = \prod_{k=1}^{n-1}k!\cdot\frac{\det(e^{x_iy_j})_{1 \le i,j \le n}}{\Delta(x)\Delta(y)} \label{eq:00.2}
\ee

\noindent (where $x_i$, $y_i$ are the eigenvalues of the Hermitian matrices $X$, $Y$, respectively, and, say, $\Delta(x) = \prod_{i<j}(x_i-x_j)$ is the Vandermonde determinant), to integrate out the unitary parts. The result is the eigenvalue measure with a density $P(\vec x_1, \dots, \vec x_m)$,

\be
P \sim  \prod_{k=1}^m e^{-\left(\frac{1}{1-c_{k-1}^2} + \frac{c_k^2}{1-c_k^2}\right)\sum_{i=1}^nx_{k,i}^2}\prod_{k=1}^{m-1}\det\left(e^{\frac{2c_k}{1-c_k^2}x_{k,i}x_{k+1,i}}\right)\Delta(x_1)\Delta(x_m).  \label{eq:00.3}
\ee

\noindent It is shown in~\cite{EyMe, TW2} that for a chain of coupled matrices with probability density of this type the correlation functions can be expressed as block determinants whose entries are matrix kernels evaluated at various points, generalizing Dyson's expression for the correlation functions of a single matrix. In the case considered here the matrix kernel 

\be
K(x, y) = \left(K_{ij}(x, y)\right)_{i,j=1}^m
\ee

\noindent is the extended Hermite kernel with entries

\be
K_{ij}(x, y) = \left\{ \begin{array}{lr} \sum_{k=0}^{n-1}e^{k(t_i-t_j)}\varphi_k(x)\varphi_k(y) & \text{if }i \ge j,  \\  \ & \ \\ -\sum_{k=n}^{\infty}e^{k(t_i-t_j)}\varphi_k(x)\varphi_k(y) & \text{if }i < j. \end{array} \right.  \label{eq:00.4}
\ee 

\noindent where $\varphi_k(x) = p_k(x)e^{-x^2/2}$ are the harmonic oscillator eigenfunctions and $p_k$ are the normalized Hermite polynomials. The actual kernel one is interested in is then

\be
K_{ij}^J(x,y) = K_{ij}(x,y)\chi_{J_j}(y),
\ee

\noindent where $\chi_{J_j}$ is the characteristic function of the set $J_j$, since the probability that for each $k$ no eigenvalue lies in $J_k$ at time $t_k$ is equal to $\det(I-K^J)$. 
\par We restrict ourselves to the case of only two coupled matrices, $m=2$, which we study in detail. Generalizations to the chain of several coupled matrices are quite feasible~\cite{EyMe, BeEyHa, TW2}, but we defer them to a future work.
\par Plan of the paper is as follows. In section 2 we show the normalization of the two-matrix integral and its one-matrix limit. Section \ref{sec:results} contains the description of the results. It connects the rest of the paper. In section \ref{sec:Toda}, using a version of the ASvM method~\cite{ASvM, AvM1}, we derive PDE for the joint distribution of two coupled Gaussian matrices from Toda lattice hierarchy. We use two copies of the AKNS system (see, e.g.~\cite{Newell}) together with the 2-Toda (Liouville) equation and find PDE different from the one of~\cite{AvM1}, and also recover the AvM equation~\cite{AvM1}. In contrast with it, our PDE have nontrivilal one-matrix limit. In section \ref{sec:4} we start from the matrix kernel and the line of TW~\cite{TW2}, but then deflect from it to derive our version of first-order system of PDE of TW type, which employs matrix analogs of inner products $u$ and $w$~\cite{TW1}, shown in~\cite{IR1, UniUE} to be universally related with the ratios of the corresponding {\it 1-Toda} $\tau$-functions, i.e. one-matrix integrals of consecutive sizes. In the next section we derive matrix recursion relations for coupled case analogous to those in~\cite{UniUE} for one-matrix unitary ensembles. In section \ref{sec:Xsys} we transform the system of section \ref{sec:4}, find many first integrals and solve it, obtaining coupled analogs of Painlev\'e IV equation~\cite{TW1}. In section \ref{sec:comp} we match a certain combination of our TW-type equations from section \ref{sec:Xsys} with a PDE from section \ref{sec:Toda} and find new correspondences among the quantities involved, in particular the relations between derivatives with respect to the coupling constant $c$ arising from the Toda lattice approach (\cite{AvM1} or section \ref{sec:Toda} here) and certain commutators involving the non-diagonal matrix elements of the matrix resolvent kernel~\cite{TW2}. In section \ref{sec:fin4eq} we show how the system of Painlev\'e IV-like PDE, the main result of section \ref{sec:Xsys}, can be reduced to several PDE for the joint probability of largest eigenvalues wihtout auxiliary variables and with derivatives only w.r.t.~the spectral endpoints of both matrices (i.e.~without derivatives w.r.t. $c$). The last section is devoted to the conclusions.

\section{Matrix integral}   \label{sec:MatrInt}

We study coupled Gaussian Hermitian unitary invariant ensemble, i.e.~consider two-matrix integral\footnote{Note that here, in the abstract two-matrix model, the normalization is different from the one for the Dyson process above. We will switch to the ``physical" normalization later.} 

\be
\tau_n = \int\int e^{-\text{Tr}M_1^2 - \text{Tr}M_2^2 + 2c\text{Tr}M_1M_2}dM_1dM_2.   \label{eq:0.1}
\ee

\noindent Integrating out the ``eigenvector" components of the matrices in the standard way, using Harish-Chandra/Itzykson-Zuber formula, and then transforming the integral, using the antisymmetry properties of the HCIZ determinant, see e.g.~\cite{AvM1, TW2}, one arrives at the formula for the joint eigenvalue distribution for this ensemble -- the probability that all eigenvalues of the first matrix lie in a set $J_1 \subset \R$, while all eigenvalues of the second matrix lie in a set $J_2 \subset \R$:


\be
\tau_n^{J_1, J_2} \sim \left(\frac{1}{c}\right)^{n(n-1)/2}\prod_{i=1}^n \int_{J_1}dx_i \prod_{i=1}^n \int_{J_2}dy_i e^{-\sum_1^n x_i^2 - \sum_1^n y_i^2 + 2c\sum_1^n x_iy_i} \Delta(x)\Delta(y),  \label{eq:0.2}
\ee

\noindent where $\Delta(x)$ is the Vandermonde determinant over $x_i$. For brevity, we will denote $\tau_n^{J_1, J_2}$ by just $\tau_n^J$.


The coupling constant $c$ varies in the range $0 \le c < 1$, so that the integral converges and the physically identical domain $-1 < c \le 0$ is excluded.



\par First consider the case $J_1 = (-\infty, \xi_1)$, $J_2 = (-\infty, \infty)$. Then change the variables $y_i \to \tilde y_i = y_i - cx_i$ to get


$$ 
\tau_n^J(\xi_1) \sim \left(\frac{1}{c}\right)^{n(n-1)/2}\prod_{i=1}^n \int_{-\infty}^{\xi_1}dx_i \prod_{i=1}^n \int_{-\infty}^{\infty}d\tilde y_i e^{-(1-c^2)\sum_1^n x_i^2 - \sum_1^n \tilde y_i^2} \Delta(x)\Delta(\tilde y + cx),  
$$ 

\noindent $\Delta(\tilde y + cx) = \prod_{i<j}(\tilde y_i - \tilde y_j + c(x_i - x_j))$. Due to the antisymmetry of the Vandermonde, the integral over new $\tilde y$ variables decouples and turns into just a multiple Gaussian integral:

$$ 
\prod_{i=1}^n \int_{-\infty}^{\infty}d\tilde y_i e^{- \sum_1^n \tilde y_i^2} \Delta(\tilde y + cx) = c^{\frac{n(n-1)}{2}} \Delta(x)\prod_{i=1}^n \int_{-\infty}^{\infty}d\tilde y_i e^{- \sum_1^n \tilde y_i^2}.  
$$

\noindent Therefore 


$$
\tau_n^J(\xi_1) \sim \prod_{i=1}^n \int_{-\infty}^{\xi_1}dx_i e^{-(1-c^2)\sum_1^n x_i^2} \Delta^2(x).     
$$

\noindent Let $\gamma = 1-c^2$, $\tilde x_i = \gamma x_i$, then


$$
\tau_n^J(\xi_1) \sim \prod_{i=1}^n \int_{-\infty}^{\gamma^{1/2}\xi_1}\frac{d\tilde x_i}{\gamma^{1/2}} e^{-\sum_1^n \tilde x_i^2} \cdot\frac{\Delta^2(\tilde x)}{\gamma^{n(n-1)/2}} = \frac{1}{\gamma^{n^2/2}}\prod_{i=1}^n \int_{-\infty}^{\gamma^{1/2}\xi_1}dx_i e^{-\sum_1^n x_i^2} \Delta^2(x),  
$$

\noindent i.e. such an integral is just a {\it renormalized} one-matrix largest eigenvalue probability:

\be
\tau_n^J(\xi_1) = \frac{1}{\gamma^{n^2/2}}\tau_n^{1-matrix}(\gamma^{1/2}\xi_1).   \label{eq:0.7}
\ee

\noindent For $\tau$-ratios of 2-matrix integrals of consecutive matrix sizes one gets therefore:

\be
U_n(\xi_1) = \frac{\tau_{n+1}^J(\xi_1)}{\tau_n^J(\xi_1)} = \frac{1}{\gamma^{n+1/2}}U_n^{1-matrix}(\gamma^{1/2}\xi_1),   \label{eq:0.8a}
\ee

\be
W_n(\xi_1) = \frac{\tau_{n-1}^J(\xi_1)}{\tau_n^J(\xi_1)} = \gamma^{n-1/2}W_n^{1-matrix}(\gamma^{1/2}\xi_1).  \label{eq:0.8b}
\ee

\noindent Thus, e.g.

\be
F_n(\xi_1) = U_nW_n(\xi_1) = \frac{1}{\gamma}F_n^{1-matrix}(\gamma^{1/2}\xi_1).   \label{eq:0.9}
\ee

\noindent In the limit when $\xi_1 \to \infty$, we get the normalization for the defined above matrix integrals over the whole domain (see, e.g.~\cite{Me04} for the last formula in eq.~(\ref{eq:0.10}) below):

\be
\tau_n = \frac{\tau_n^{1-matrix}}{(1-c^2)^{n^2/2}},   \ \ \ \ \ \tau_n^{1-matrix} = \frac{\pi^{n/2}}{2^{n(n-1)/2}}\prod_{j=1}^{n-1}j!,     \label{eq:0.10}
\ee

\be
\frac{\tau_{n+1}\tau_{n-1}}{\tau_n^2} = \frac{n}{2(1-c^2)}.   \label{eq:0.11}
\ee

The above matrix integral is a $\tau$-function of 2-Toda integrable hierarchy~\cite{AvM1} if one considers its modification by introducing two infinite sets of ``times" or coupling parameters $t$ and $s$,

$$
\tau_n(t, s) = \int\int e^{-\text{Tr}M_1^2 - \text{Tr}M_2^2 + 2c\text{Tr}M_1M_2 + \sum_{k=1}^{\infty}t_k\text{Tr}M_1^k - \sum_{k=1}^{\infty}s_k\text{Tr}M_2^k}dM_1dM_2.
$$

\section{Description of the results}   \label{sec:results}

We first apply a version of approach of~\cite{ASvM, AvM1}. We use 2-Toda bilinear identity, which supplies five simplest integrable PDE of its series -- the 2-Toda (or Liouville) equation and two copies of AKNS system. Then the Virasoro constraints give us expressions of ``time" derivatives in terms of spectral endpoint derivatives, which allow us to obtain our first system of five scalar PDE for joint gap probabilities of two coupled GUE w.r.t.~the spectral endpoints:

\begin{theorem}   \label{thm:1}

The logarithm $T$ of the joint spacing probability of two coupled Gaussian matrices or the joint two-time probability for the Gaussian Dyson process, together with the auxiliary functions $U$ and $W$, satisfies the following system of PDE:

$$
\A\tilde\A T = 4c(UW - n/2), \eqno(\ref{eq:00})
$$

$$
\A^2 U = -2\A_0 U - 2\A^2T\cdot U,  \eqno(\ref{eq:+t})
$$

$$
\A^2 W = 2\A_0 W - 2\A^2T\cdot W,  \eqno(\ref{eq:-t})
$$ 

$$
\tilde\A^2 U = -2\tilde\A_0 U - 2\tilde\A^2T\cdot U,  \eqno(\ref{eq:+s})
$$ 

$$
\tilde\A^2 W = 2\tilde\A_0 W - 2\tilde\A^2T\cdot W.  \eqno(\ref{eq:-s})
$$

\end{theorem}

\ni Here we denoted

\be
\A =  \B_{-1} + c\tilde\B_{-1}, \ \ \ \tilde\A = c\B_{-1} + \tilde\B_{-1},   
\ee

$$
\A_0 = \B_0 + c^2\tilde\B_0 - \D_c, \ \ \ \tilde\A_0 = \tilde\B_0 + c^2\B_0 - \D_c, \ \ \ \D_c = (1-c^2)c\prt_c,
$$

$$
\B_k = \sum_{i \in \text{endpoints of } M_1} a_i^{k+1}\frac{\prt}{\prt a_i}, \ \ \ \ \tilde\B_k = \sum_{i \in \text{endpoints of } M_2} a_i^{k+1}\frac{\prt}{\prt a_i},
$$

\noindent and 

$$
T = \ln\tau_n^J, \ \ \ \ U = U_n(1-c^2)^{n+1/2}, \ \ \ \ W = W_n(1-c^2)^{-(n-1/2)}, \ \ \ \ U_n = \frac{\tau_{n+1}^J}{\tau_n^J}, \ \ \ \ W_n = \frac{\tau_{n-1}^J}{\tau_n^J}.
$$ 

\noindent When all endpoints, say, of the second matrix go to infinity, the coupled matrix integral becomes the (renormalized) corresponding one-matrix integral, see the previous section. Our system of equations reduces in this limit to the three equations for single Gaussian matrices obtained in~\cite{IR1}. The known AvM equation~\cite{AvM1} for coupled Gaussian matrices, however, becomes trivial in this limit. So our system contains some additional information.
\par Solving this system in general, we find two higher (fifth) order PDE satisfied by $T$:

\begin{theorem}   \label{thm:genT}

The logarithm of the joint spacing probability for two coupled Gaussian matrices (or two time-point distribution for the Dyson process) satisfies the two higher-order PDE below:

\be
\A\left(-\A^2F - 4\A_0T - 4\A^2TF + \frac{(\A F)^2 - G^2}{F} \right) = -\tilde\A\left(\frac{1}{2c}(\A^2T)^2 + \frac{1}{c}(\A_0^2 + 2\A_0)T\right),  \label{eq:T1}
\ee

\be
\tilde\A\left(-\tilde\A^2F - 4\tilde\A_0T - 4\tilde\A^2TF + \frac{(\tilde\A F)^2 - \tilde G^2}{F} \right) = -\A\left(\frac{1}{2c}(\tilde\A^2T)^2 + \frac{1}{c}(\tilde\A_0^2 + 2\tilde\A_0)T\right). \label{eq:T2}
\ee

\noindent Expressions for $F \equiv UW$, $G \equiv W\A U - U\A W$, $\tilde G = W\tilde \A U - U\tilde \A WG$ in terms of $T$ read as

\be
F = \frac{1}{4c}\A\tilde\A T + \frac{n}{2}, \ \ \ G = \A T - \frac{1}{2c}\tilde\A\A_0T, \ \ \ \tilde G = \tilde\A T - \frac{1}{2c}\A\tilde\A_0T,  \label{eq:GG2}
\ee

\ni so

$$
\frac{(\A F)^2 - G^2}{F} = \frac{(\A^2\tilde\A T)^2 - 4(\tilde\A\A_0T - 2c\A T)^2}{4c(\A\tilde\A T + 2cn)}, \ \ \  \frac{(\tilde\A F)^2 - \tilde G^2}{F} = \frac{(\A\tilde\A^2 T)^2 - 4(\A\tilde\A_0T - 2c\tilde\A T)^2}{4c(\A\tilde\A T + 2cn)}
$$

\end{theorem}

\ni Besides, we recover in our approach also the known third-order AvM equation~\cite{AvM1}.
\par Further on we consider in greater details the case of joint largest eigenvalue distribution for two coupled GUE. In this case, a simpler fourth-order PDE can be derived from the system of the first theorem, which we later compare with our equations derived from TW approach of matrix kernel Fredholm determinant.
For this purpose, it is convenient to rewrite everything in a different way. Introduce notations: 

\be
\D_+ = \prt_{\xi_1} + \prt_{\xi_2}, \ \ \ \D_- = \prt_{\xi_1} - \prt_{\xi_2}, \ \ \ \xi_+ = \xi_1 + \xi_2, \ \ \ \xi_- = \xi_1 - \xi_2,
\ee

\noindent and 

\be
\sigma = \frac{1-c}{1+c}.
\ee

\ni The expression for $F$ above now reads

\be
F = \frac{(\D_+^2 - \sigma^2\D_-^2) T}{4(1-\sigma^2)} + \frac{n}{2}. \label{eq:0+}  
\ee

\ni Let 

$$
G_+ = W\D_+ U - U\D_+ W, \ \ \ G_- =  W\D_- U - U\D_- W.   
$$

\ni The last quantities have expressions in terms of $T$, corresponding to (\ref{eq:GG2}),

\be
G_+ = 1/2\D_+ T - 1/4\xi_+\D_+^2T - 1/4\xi_-\D_+\D_-T + \frac{1}{2c}\D_c\D_+T - 1/2\xi_+\sigma^2\frac{(\D_+^2 - \D_-^2)T}{(1-\sigma^2)},    \label{eq:G+} 
\ee
 
\be
G_- = 1/2\D_- T - 1/4\xi_-\D_-^2T - 1/4\xi_+\D_+\D_-T - \frac{1}{2c}\D_c\D_-T - 1/2\xi_-\frac{(\D_+^2 - \D_-^2)T}{(1-\sigma^2)}.    \label{eq:G-} 
\ee
 
\ni Then we have

\begin{theorem}   \label{thm:F4T}

The logarithm of the joint largest eigenvalue probability for two coupled Gaussian matrices (or two time-point distribution for the Dyson process) satisfies the fourth-order PDE:

$$
2F\D_+\D_-F - \D_+F\D_-F + G_+G_- + 2F(\xi_-G_+ + \xi_+G_-) + 8\D_+\D_-T\cdot F^2 = 0,  \eqno(\ref{eq:2})
$$

\ni where $F$ is given by formula (\ref{eq:0+}), while $G_+$, $G_-$ are given by formulas (\ref{eq:G+}), (\ref{eq:G-}), respectively, in terms of $T = \ln\tau_n^J$.

\end{theorem}

\ni Its one-matrix limit, i.e. the limit as $\xi_1 \to \infty$ or $\xi_2 \to \infty$, is a combination of the 3rd order equation in the derivative of $T$~\cite{TW1}, which gives Painlev\'e IV after integration, and Painlev\'e IV itself. This is again in contrast to the equation (\ref{eq:AvM}), trivial in this limit. 
\par It turns out that we can get more complete information from the matrix kernel approach of~\cite{TW2}. First, we derive the large system below, which turns out to be highly redundant. It contains $2\times 2$ matrix variables analogous to the scalar variables of~\cite{TW1}. Here the main variables are the entries of the matrix $r$, $r_{ij} = R_{ij}(\xi_i, \xi_j)$, which contains values of the matrix resolvent kernel $R$ (for Fredholm matrix operator $K^J$) at the spectral endpoints, its diagonal entries are the first derivatives of the previously introduced $T \sim \ln\Pr$,

$$
r_t \equiv \text{Tr}r = \D_+T, \ \ \ \ r_3 \equiv \text{Tr}(\sigma_3 r) = \D_-T,
$$

\ni while the anti-diagonal entries are expressed in terms of the matrix

\be
A = \frac{[\sigma_3, r - 1/2[\sigma, r]]}{1 - \sigma^2}, \ \ \ \ \ \ \tilde A = (\sigma_3\sigma)A,   \label{eq:Adef}
\ee

\noindent where we have introduced the matrix $\sigma$,

\be
\sigma = \frac{1-c}{1+c}\sigma_3.
\ee

\ni As a consequence of the system below, the scalar $A^2$ ($A^2 \sim r_a^2$, square of the anti-diagonal part of $r$) is equal:

\be
A^2 = \frac{\D_+ r_t - \D_-r_3}{1-\sigma^2} = \frac{\D_+^2\ln\tau_n^J - \D_-^2\ln\tau_n^J}{1-\sigma^2} = \frac{4\prt^2_{\xi_1\xi_2}\ln\tau_n^J}{1-\sigma^2}.   \label{eq:A^2}
\ee

\begin{theorem}   \label{thm:QPsys}

The joint largest eigenvalue probabilities for the ensemble of two coupled Gaussian matrices and related auxiliary variables satisfy the following $2\times 2$ matrix system of first order PDE:

$$
\D_+r = -\frac{1}{2}\hat q\hat{\tilde p}(I-\sigma) - \frac{1}{2}(I+\sigma)\hat p\hat{\tilde q} - [\sigma\xi, r], \eqno(\ref{eq:r+})
$$

$$
(\sigma_3\sigma)\D_-r = \frac{1}{2}\hat q\hat{\tilde p}(I-\sigma) - \frac{1}{2}(I+\sigma)\hat p\hat{\tilde q} - [\xi, r], \eqno(\ref{eq:r-})
$$

$$
\D_+\hat q = -(\xi + \tilde A)\hat q + \hat p\hat U,   \eqno(\ref{eq:q+})
$$

$$
\D_-\hat q = (-\sigma_3\xi + A)\hat q + \sigma_3\hat p\hat U,   \eqno(\ref{eq:q-})
$$

$$
\D_+\hat p = (\xi + \tilde A)\hat p - \hat q\hat W,   \eqno(\ref{eq:p+})
$$

$$
\D_-\hat p = (\sigma_3\xi + A)\hat p - \sigma_3\hat q\hat W,   \eqno(\ref{eq:q-})
$$
  
$$
\D_+\hat{\tilde q} = -\hat{\tilde q}(\xi + \tilde A) + \hat U\hat{\tilde p},   \eqno(\ref{eq:tq+})
$$

$$
\D_-\hat{\tilde q} = -\hat{\tilde q}(\sigma_3\xi + A)\hat q + \hat U\hat{\tilde p}\sigma_3,   \eqno(\ref{eq:tq-})
$$

$$
\D_+\hat{\tilde p} = \hat{\tilde p}(\xi + \tilde A) - \hat W\hat{\tilde q},   \eqno(\ref{eq:tp+})
$$

$$
\D_-\hat{\tilde p} = \hat{\tilde p}(\sigma_3\xi - A) - \hat W\hat{\tilde q}\sigma_3,   \eqno(\ref{eq:tp-})
$$

\be
\D_+\hat U = \hat{\tilde q}\hat q,  \label{eq:U+}
\ee

\be
\D_-\hat U = \hat{\tilde q}\sigma_3\hat q,  \label{eq:U-}
\ee

\be
\D_+\hat W = -\hat{\tilde p}\hat p,  \label{eq:W+}
\ee

\be
\D_-\hat W = -\hat{\tilde p}\sigma_3\hat p.  \label{eq:W-}
\ee

\end{theorem}

\ni The other matrices in the above system of equations are degenerate (determinant zero). 
\par Besides, there are simple universal relations here which are direct analogs of the corresponding relations for the one-matrix case~\cite{UniUE}, involving matrix analogs $\bar u\equiv \bar u_n = (\varphi_n, (I-K_n^J)\varphi_n)$ and $\bar w\equiv \bar w_n = (\varphi_{n-1}, (I-K_n^J)\varphi_{n-1})$ of inner products of~\cite{TW1}, see sections \ref{sec:4} and \ref{sec:unirel}. 



  
Then, with $\Theta$ being the matrix with all elements equal unity (as in~\cite{TW2}) and $\hat{\bar w}_n = (I-\sigma)\bar w_n(I+\sigma)/(1-\sigma^2)$, define 

$$
\hat{\bar U} = \Theta - \Theta \bar u_n\Theta,   \ \ \ \ \ \hat{\bar W} = \Theta + \Theta \hat{\bar w}_n\Theta.
$$

\ni Then in the Gaussian case $\hat U = \sqrt{n/2}\hat{\bar U}$ and $\hat W = \sqrt{n/2}\hat{\bar W}$.

\begin{theorem}    \label{thm:unirel}

The correspondence between TW and ASvM variables found for 1-matrix case in~\cite{IR1, UniUE}, is modified for the coupled  case\footnote{This is true in general for any matrix kernel of the form (\ref{eq:4.3}) like it was for the 1-matrix case~\cite{UniUE} since we have not used the differentiation formulas for the harmonic oscillator wavefunctions here.} as follows:
 
\be
\frac{\tau_{n+1}^J/\tau_{n+1}}{\tau_n^J/\tau_n} = \det(I - \Theta \bar u_n) = 1 - \text{Tr}(\Theta \bar u_n),
\ee

\be
\frac{\tau_{n-1}^J/\tau_{n-1}}{\tau_n^J/\tau_n} = \det(I + \Theta\hat{\bar w}_n) = 1 + \text{Tr}(\Theta\hat{\bar w}_n),
\ee

therefore

\be
 \text{Tr} \hat{\bar U} = 2\;\frac{\tau_{n+1}^J/\tau_{n+1}}{\tau_n^J/\tau_n}, \ \ \ \ \ \ \text{Tr}\hat{\bar W}  = 2\;\frac{\tau_{n-1}^J/\tau_{n-1}}{\tau_n^J/\tau_n}.
\ee

\end{theorem}

\ni In fact, one has also a remarkable result:

\begin{lemma}   \label{lem:tr}



\be
 \text{Tr} (\hat{\bar U}\hat{\bar W}) = \text{Tr}\hat{\bar U}\cdot \text{Tr}\hat{\bar W}.
\ee

\end{lemma}

\ni All this plays a role in reducing the effective number of independent scalar variables and making the matrix system redundant. To solve it, we first reduced it to another matrix system (still highly redundant for the above reasons). Let 

\be
X_+ = \hat q\hat{\tilde p} + \hat p\hat{\tilde q}, \ \ \ \ \ \ X_- = \hat q\hat{\tilde p} - \hat p\hat{\tilde q},    \label{eq:X+-}
\ee

\be
\T = 2(\hat p\hat U\hat{\tilde p} - \hat q\hat W\hat{\tilde q}), \ \ \ \ \ \ G = 2(\hat p\hat U\hat{\tilde p} + \hat q\hat W\hat{\tilde q}),    \label{eq:PhiG}
\ee 

\noindent then 

\begin{theorem}   \label{thm:Xsys}

The following closed system of ten PDE is a direct consequence of the system of fourteen PDE obtained in section \ref{sec:4}:

$$
\D_+ r = -1/2X_+ - 1/4[\sigma, X_+] + 1/4\{\sigma, X_-\} - [\sigma\xi, r],     \eqno(\ref{eq:r+})
$$

$$
(\sigma_3\sigma)\D_- r = 1/2X_- - 1/4\{\sigma, X_+\} + 1/4[\sigma, X_-]   - [\xi, r],      \eqno(\ref{eq:r-})
$$

\noindent and

\be
\D_+ X_+ = \T - [\xi + \tilde A, X_-],     \label{eq:D+X+}
\ee

\be
\D_+ X_- = - [\xi + \tilde A, X_+],       \label{eq:D+X-}
\ee

\be
\D_+ \T = 3X_+^2 - 8nX_+ + X_-^2 + \{\xi + \tilde A, G\},      \label{eq:D+T}
\ee

\be
\D_+ G = \{\xi + \tilde A, \T\} + [X_+, X_-],      \label{eq:D+G}
\ee

\be
\D_- X_+ = 1/2\{\sigma_3, \T\} - [\sigma_3\xi, X_-] + [A, X_+],     \label{eq:D-X+}   
\ee

\be
\D_- X_- = 1/2[\sigma_3, G] - [\sigma_3\xi, X_+] + [A, X_-],      \label{eq:D-X-}
\ee

\be
\D_- \T = \{\sigma_3, X_+^2 - 4nX_+ + X_-^2\} + X_+\sigma_3X_+ - X_-\sigma_3X_- + \{\sigma_3\xi, G\} + [A, \T],     \label{eq:D-T}
\ee

\be
\D_- G = [\sigma_3, \{X_+, X_-\} - 4nX_-] + X_+\sigma_3X_- - X_-\sigma_3X_+ + \{\sigma_3\xi, \T\}+[A, G],      \label{eq:D-G}
\ee

\end{theorem}

\begin{proof}

The first two equations above are just the first two equations of theorem \ref{thm:QPsys}, getting the others is straightforward expanding of their left-hand sides by definitions (\ref{eq:X+-}) and (\ref{eq:PhiG}) and then applying the rest of equations of theorem \ref{thm:QPsys} together with definitions (\ref{eq:Adef}).

\end{proof}

\ni The system has a number of first integrals some of which are matrix while the other involve only some of the scalar components, thus forcing one to split the system into scalar parts eventually. As the main consequence of the system, we found a complete system of independent third-order PDE satisfied by joint largest eigenvalue distributions for two coupled GUE:

\begin{corollary}     \label{cor:main}


$$
\Pe_x \equiv \D_+X_t\D_-X_t - 2X_tX_3\hat F - G_tG_3 = 0, \eqno(\ref{eq:x})
$$

$$
D_x \equiv 4\sigma^2(\D_+A^2\D_-A^2 - 2X_3(A^2)^2) - A_+A_- = 0, \eqno(\ref{eq:Ax})
$$

\be
2\sigma^2A^2\Pe_t = -2\sigma^2A^2\Pe_3 = -\hat FD_t = \sigma^2\hat FD_3,   \label{eq:+-}
\ee

$$
\hat F(\D_+A^2A_- - \D_-A^2A_+) = A^2(G_3\D_+X_t - G_t\D_-X_t), \eqno(\ref{eq:a})
$$

\noindent where

\be
\Pe_t =  (\D_+X_t)^2 - \hat F(2X_3^2 + J) - G_t^2,  \ \ \ \Pe_3 = (\D_-X_t)^2 - \hat F(2X_3^2 + J) - G_3^2,     \label{eq:PX}
\ee

\be
D_t = 4\sigma^2(\D_+A^2)^2 + 2\sigma^2A^2J - A_+^2, \ \ \ D_3 = 4\sigma^2(\D_-A^2)^2 + 2A^2J - A_-^2.    \label{eq:PA}
\ee

$$
X_t = -2\D_+ r_t - 2\sigma^2A^2 = -2\D_-r_3 - 2A^2 = -2\frac{(\D_+^2 - \sigma^2\D_-^2)\ln\tau_n^J}{1-\sigma^2}, 
$$

$$
X_3 = -2\D_+ r_3 = \D_-r_t = -2\D_+\D_-\ln\tau_n^J,
$$

$$
\hat F = X_t - 4n, \ \ \ \ J = X_t^2 - X_3^2 - 4\sigma^2(A^2)^2,
$$

$$
D_+ = \D_+ A^2, \ \ \ \ D_- = \D_-A^2, 
$$

$$
G_t = H_t + A_+, \ \ \ G_3 = H_3 + A_-, \ \ \ \ H_t = 4r_t - 2\xi_+\D_+ r_t + \xi_-X_3,  \ \ \ \ H_3 = 4r_3 - 2\xi_-\D_- r_3 + \xi_+X_3.
$$

\be
A_+ = \text{tr}\{\tilde A, (X_+)_a\}, \ \ \ \ A_- =  - \text{tr}\{A, (X_-)_a\},    \label{eq:A+A-}
\ee

\ni ($\{,\}$ is an anti-commutator, subscript 'a' means the anti-diagonal part of a matrix).

\end{corollary}

\noindent The last four expressions clearly have the three-term structure resembling that of Painlev\'e IV equation for one-matrix Gaussian ensemble, as do the left-hand sides of (\ref{eq:x}) and (\ref{eq:Ax}). And indeed the limit of $\Pe_t$, $\Pe_3$ and $\Pe_x$ as either $\xi_2 \to \infty$ or $\xi_1 \to \infty$ is the Painlev\'e IV equation itself. The one-matrix integrals for GUE largest eigenvalue probabilities, which are the appropriate solutions of the limiting Painlev\'e IV equations, should in fact be considered as {\it boundary conditions} for the corresponding PDEs of the last corollary. 
\par This nice form of the equations is reached at the expense of still keeping two auxiliary variables $A_+$ and $A_-$ in them. They can be explicitly eliminated giving a smallest complete system of four independent third-order PDE for $T$, which, however, lack the nice Painlev\'e IV-like structure of the previous ones, see section \ref{sec:fin4eq}. 
\par Comparison of a combination of 4th-order PDEs of theorem \ref{theorem:4ord} and 3rd-order PDEs of corollary \ref{cor:main} obtained from the variant of TW system in section \ref{sec:Xsys} with the 4th-order PDE of theorem \ref{thm:F4T} above from our version of AvM approach, gives the correspondences of the variables of both approaches, among them -- some new ones, specific for the case of coupled matrices:

\be
\sigma_3[r_a, \D_+ r_a] = -2c\prt_cr_t, \ \ \ \ \sigma_3[r_a, \D_- r_a] = 2c\prt_cr_3.   \label{eq:ccom}
\ee

\ni These last two new important relations show that the main additional variables appearing in the TW approach to coupled (Gaussian) matrices -- the commutators on the left-hand side -- directly correspond to the main new objects appearing in the ASvM approach to the same problem -- the derivatives w.r.t.~the coupling $c$ (or w.r.t.~time $t$, recall $c = e^{-t}$) of the log-derivatives of the joint largest eigenvalue probability.

\section{System of PDE from Toda lattice hierarchy}    \label{sec:Toda}

Along with~\cite{AvM1}, we first consider the following 2-Toda $\tau$-function:

\be
\tau_n(t, s) = \prod_{i=1}^n \int_{J_1}dx_i \int_{J_2}dy_i e^{-x_i^2 - y_i^2 + 2cx_iy_i + \sum_{k=1}^{\infty}t_kx_i^k - s_ky_i^k} \Delta(x)\Delta(y) \sim c^{\frac{n(n-1)}{2}}\tau_n(t, s)  \label{eq:init}
\ee

\ni of section \ref{sec:MatrInt}. From the bilinear identity for the 2-Toda hierarchy~\cite{UT84}, see also~\cite{AvM1},

$$
\oint_{z=\infty}\tau_n(t-[z^{-1}], s)\tau_{m+1}(t'+[z^{-1}], s')e^{\sum_1^{\infty}(t_k-t'_k)z^k} z^{n-m-1}dz = 
$$

$$
= \oint_{z=0}\tau_{n+1}(t, s-[z])\tau_m(t', s'+[z])e^{\sum_1^{\infty}(s_k-s'_k)z^{-k}} z^{n-m-1}dz,      
$$

\ni where $[z] = (z, z^2/2, z^3/3, \dots)$ -- infinite vector, for $m=n\pm 1$ one obtains the series of PDE:

\be
-\frac{\prt^2\ln\tau_n}{\prt t_k\prt s_1} = \sum_{i=0}^{k-1} \frac{p_i(-\tilde\prt_t)\tau_{n-1}\cdot p_{k-1-i}(\tilde\prt_t)\tau_{n+1}}{\tau_n^2}
\ee

\ni and

\be
-\frac{\prt^2\ln\tau_n}{\prt t_1\prt s_k} = \sum_{i=0}^{k-1} \frac{p_i(-\tilde\prt_s)\tau_{n+1}\cdot p_{k-1-i}(\tilde\prt_s)\tau_{n-1}}{\tau_n^2},
\ee

\noindent both of which give as the simplest ($k=1$) equation the 2-Toda or Liouville equation in terms of $\tau$-functions:

\be
-\frac{\prt^2\ln\tau_n}{\prt t_1\prt s_1} = \frac{\tau_{n+1}\tau_{n-1}}{\tau_n^2}.
\ee

\noindent The identity for $m=n$ gives another two series of PDE:

\be
\frac{\prt}{\prt t_k}\ln\frac{\tau_{n+1}}{\tau_n} = \sum_{i=0}^{k} \frac{p_i(-\tilde\prt_t)\tau_{n}\cdot p_{k-i}(\tilde\prt_t)\tau_{n+1}}{\tau_n\tau_{n+1}},
\ee

\be
-\frac{\prt}{\prt s_k}\ln\frac{\tau_{n+1}}{\tau_n} = \sum_{i=0}^{k} \frac{p_i(-\tilde\prt_s)\tau_{n+1}\cdot p_{k-i}(\tilde\prt_s)\tau_{n}}{\tau_n\tau_{n+1}},
\ee

\ni which give the simplest nontrivial equations at $k=2$, namely two copies of AKNS system:

\be
\frac{\partial U_n}{\partial t_2} = \frac{\partial^2 U_n}{\partial t_1^2} + 2\frac{\partial^2\ln\tau_n}{\partial t_1^2}U_n, \ \ \ \ \ \ -\frac{\partial W_n}{\partial t_2} = \frac{\partial^2 W_n}{\partial t_1^2} + 2\frac{\partial^2\ln\tau_n}{\partial t_1^2}W_n, 
\ee

\be
-\frac{\partial U_n}{\partial s_2} = \frac{\partial^2 U_n}{\partial s_1^2} + 2\frac{\partial^2\ln\tau_n}{\partial s_1^2}U_n, \ \ \ \ \ \ \frac{\partial W_n}{\partial s_2} = \frac{\partial^2 W_n}{\partial s_1^2} + 2\frac{\partial^2\ln\tau_n}{\partial s_1^2}W_n. 
\ee

\noindent Then we use the Virasoro constraints~\cite{ASvM, AvM7, AvM1, IR1}, connecting the derivatives w.r.t.~the spectral endpoints with the time derivatives of the integral $\tau_n^J(t, s)$ (\ref{eq:init}). We need only the $k=-1$ constraints:

\be
\B_{-1}\tau_n^J(t, s) = \V_{-1}\tau_n^J(t, s) = \left\{-2\frac{\partial}{\partial t_1} + nt_1 + \sum_{l=2}^{\infty}lt_l\frac{\partial}{\partial t_{l-1}} - 2c\frac{\partial}{\partial s_1}\right\} \tau_n^J(t, s),   
\ee

\be
\tilde\B_{-1}\tau_n^J(t, s) = \tilde\V_{-1}\tau_n^J(t, s) = \left\{2\frac{\partial}{\partial s_1} - ns_1 + \sum_{l=2}^{\infty}ls_l\frac{\partial}{\partial s_{l-1}} + 2c\frac{\partial}{\partial t_1}\right\} \tau_n^J(t, s),   
\ee

\noindent and the $k=0$ constraints:

\be
\B_0\tau_n^J(t, s) = \V_0\tau_n^J(t, s) = \left\{-2\frac{\partial}{\partial t_2} + \frac{n(n+1)}{2} + \sum_{l=1}^{\infty}lt_l\frac{\partial}{\partial t_l} + c\frac{\prt}{\prt c}\right\} \tau_n^J(t, s), 
\ee

\be
\tilde\B_0\tau_n^J(t, s) = \V_0\tau_n^J(t, s) = \left\{2\frac{\partial}{\partial s_2} + \frac{n(n+1)}{2} + \sum_{l=1}^{\infty}ls_l\frac{\partial}{\partial s_l} + c\frac{\prt}{\prt c}\right\} \tau_n^J(t, s). 
\ee

\noindent So on locus $t_k = s_k = 0$ one gets:   

$$
\B_{-1}\ln\tau_n^J = -2\frac{\prt\ln\tau_n^J}{\prt t_1} - 2c\frac{\prt\ln\tau_n^J}{\prt s_1},
$$

$$
\tilde\B_{-1}\ln\tau_n^J = 2\frac{\prt\ln\tau_n^J}{\prt s_1} + 2c\frac{\prt\ln\tau_n^J}{\prt t_1},
$$

$$
\B_0\tau_n^J = -2\frac{\partial\ln\tau_n^J}{\partial t_2} + \frac{n(n+1)}{2} + c\frac{\prt}{\prt c}\ln\tau_n^J, 
$$

$$
\tilde\B_0\tau_n^J = 2\frac{\partial\ln\tau_n^J}{\partial s_2} + \frac{n(n+1)}{2} + c\frac{\prt}{\prt c}\ln\tau_n^J. 
$$

\noindent Constraints for the second derivatives follow from the above, using the commutativity of boundary and time derivatives. On the locus those we need reduce to:

$$
\B_{-1}^2\ln\tau_n^J = 4\frac{\prt^2\ln\tau_n^J}{\prt t_1^2} - 2n + 8c\frac{\prt^2\ln\tau_n^J}{\prt t_1\prt s_1} + 4c^2\frac{\prt^2\ln\tau_n^J}{\prt s_1^2},
$$

$$
\tilde\B_{-1}^2\ln\tau_n^J = 4\frac{\prt^2\ln\tau_n^J}{\prt s_1^2} - 2n + 8c\frac{\prt^2\ln\tau_n^J}{\prt t_1\prt s_1} + 4c^2\frac{\prt^2\ln\tau_n^J}{\prt t_1^2},
$$

$$
\B_{-1}\tilde\B_{-1}\ln\tau_n^J = -4(1+c^2)\frac{\prt^2\ln\tau_n^J}{\prt t_1\prt s_1} + 2cn - 4c\frac{\prt^2\ln\tau_n^J}{\prt t_1^2} - 4c\frac{\prt^2\ln\tau_n^J}{\prt s_1^2}.
$$

\noindent Expressing the time derivatives in terms of the boundary derivatives as usual in the ASvM approach~\cite{ASvM, AvM1}, we finally find for the second derivatives of $\ln\tau_n^J$:

\be
\frac{\partial^2\ln\tau_n^J}{\partial t_1^2} = \frac{(\B_{-1} + c\tilde\B_{-1})^2\ln\tau_n^J + 2(1-c^2)n}{4(1-c^2)^2},
\ee 

\be
\frac{\partial^2\ln\tau_n^J}{\partial s_1^2} = \frac{(c\B_{-1} + \tilde\B_{-1})^2\ln\tau_n^J + 2(1-c^2)n}{4(1-c^2)^2},
\ee 

\be
\frac{\partial^2\ln\tau_n^J}{\prt t_1\prt s_1} = -\frac{(\B_{-1} + c\tilde\B_{-1})(c\B_{-1} + \tilde\B_{-1})\ln\tau_n^J + 2c(1-c^2)n}{4(1-c^2)^2}.
\ee

After plugging the last expression into the 2-Toda equation, the ``boundary-Toda" equation for coupled matrices follows:

\be
(\B_{-1} + c\tilde\B_{-1})(c\B_{-1} + \tilde\B_{-1})\ln\tau_n^J = 4(1-c^2)^2U_nW_n - 2c(1-c^2)n.   \label{eq:00}
\ee

\noindent Taking the corresponding formulas for the ratios $U_n$, $W_n$, and recalling the definitions from section \ref{sec:results},


one obtains four ``boundary'' equations from two AKNS copies:

\be
\A^2U_n = -2(1-c^2)^2(\B_0 - c\prt_c)U_n - 2(1-c^2)((1+c^2)n + c^2)U_n - 2\A^2\ln\tau_n^JU_n,   \label{eq:+t}
\ee

\be
\A^2W_n = 2(1-c^2)^2(\B_0 - c\prt_c)W_n - 2(1-c^2)((1+c^2)n - 1)W_n - 2\A^2\ln\tau_n^JW_n,   \label{eq:-t}
\ee

\be
\tilde\A^2U_n = -2(1-c^2)^2(\tilde\B_0 - c\prt_c)U_n - 2(1-c^2)((1+c^2)n + c^2)U_n - 2\tilde\A^2\ln\tau_n^JU_n,   \label{eq:+s}
\ee

\be
\tilde\A^2W_n = 2(1-c^2)^2(\tilde\B_0 - c\prt_c)W_n - 2(1-c^2)((1+c^2)n - 1)W_n - 2\tilde\A^2\ln\tau_n^JW_n.   \label{eq:-s}
\ee

\noindent Now let us pass to the ``physical" Dyson process variables, namely substitute

\be
\xi \to \frac{\xi}{\sqrt{1-c^2}}, \ \ \ \ \ \ \tau_n^J \to c^{\frac{n(n-1)}{2}}\tau_n^J,    \label{eq:phys}
\ee

\noindent (if we did only the second substitution above we would get exactly the matrix integral of section 2 as $\tau_n^J$) and so

\be
U_n \to c^n U_n, \ \ \ W_n \to W_n/c^{n-1},
\ee

\noindent which entails also

\be
c\frac{\prt}{\prt c}U_n \to c^n\left(c\frac{\prt}{\prt c} + n\right)U_n, \ \ \ \ \ \ c\frac{\prt}{\prt c}W_n \to \frac{1}{c^{n-1}}\left(c\frac{\prt}{\prt c} - n + 1\right)W_n
\ee

\noindent {\it and}, importantly, change in the meaning of the partial derivative with respect to $c$ (due to the acquired dependence on $c$ of former $\xi$ under the change of variables (\ref{eq:phys})), so that one has to substitute everywhere

\be
-c\prt_c = 2c^2\prt_{\gamma} \to 2c^2\left(\prt_{\gamma} + \frac{\B_0 + \tilde\B_0}{2\gamma}\right).
\ee
 
\noindent After these transformations the previous system of PDE reads:

$$
\A\tilde\A \ln\tau_n^J = 2c(2(1-c^2)U_nW_n - n), \eqno(\ref{eq:00})
$$

$$
\A^2 U_n = -2(\B_0 + c^2\tilde\B_0)U_n + 2(1-c^2)c\prt_c U_n - 2\A^2\ln\tau_n^J U_n - 2c^2(2n+1)U_n,  \eqno(\ref{eq:+t})
$$

$$
\A^2 W_n = 2(\B_0 + c^2\tilde\B_0)W_n - 2(1-c^2)c\prt_c W_n - 2\A^2\ln\tau_n^J W_n - 2c^2(2n-1)W_n,  \eqno(\ref{eq:-t})
$$ 

$$
\tilde\A^2 U_n = -2(\tilde B_0 + c^2\B_0)U_n + 2(1-c^2)c\prt_c U_n - 2\tilde\A^2\ln\tau_n^J U_n - 2c^2(2n+1)U_n,  \eqno(\ref{eq:+s})
$$

$$
\tilde\A^2 W_n = 2(\tilde B_0 + c^2\B_0)W_n - 2(1-c^2)c\prt_c W_n - 2\tilde\A^2\ln\tau_n^J W_n - 2c^2(2n-1)W_n.  \eqno(\ref{eq:-s})
$$ 





\noindent At last, it is convenient to remove the last terms in the last four equations by passing to new functions,

$$
U = U_n(1-c^2)^{n+1/2}, \ \ \ \ W = W_n(1-c^2)^{-(n-1/2)},
$$

\noindent which also simplifies the coefficient in equation (\ref{eq:00}) and the final system becomes that of theorem \ref{thm:1}.









\par Now we try to solve this system. The procedure goes in parallel with one-matrix case (see~\cite{UniUE, Multipt1matr}). First introduce quantities

$$
F = UW, \ \ \ G = W\A U - U\A W, \ \ \ \tilde G = W\tilde \A U - U\tilde \A W,
$$


$$
G_0 = W\A_0 U - U\A_0 W, \ \ \ \tilde G_0 = W\tilde \A_0 U - U\tilde \A_0 W.  
$$

\noindent Multiplying eqs.~(\ref{eq:+t}) and (\ref{eq:+s}) by $W$ and eqs.~(\ref{eq:-t}) and (\ref{eq:-s}) by $U$ and taking the appropriate combinations to express everything in terms of the newly introduced functions, we obtain our system in the form:

\be
\A\tilde\A T = 4c(F - n/2),  \label{eq:0}
\ee

\be
\A G = -2\A_0 F,   \label{eq:AG} 
\ee

\be
\tilde\A \tilde G = -2\tilde\A_0 F,  \label{eq:tAG} 
\ee

\be
\A^2F - \frac{(\A F)^2 - G^2}{2F} = -2G_0 - 4\A^2T\cdot F,  \label{eq:AF} 
\ee

\be
\tilde\A^2F - \frac{(\tilde\A F)^2 - \tilde G^2}{2F} = -2\tilde G_0 - 4\tilde\A^2T\cdot F.  \label{eq:tAF} 
\ee

\noindent We will need commutation relations among our operators:





\be
[\A, \A_0] = \A, \ \ \ \ [\A, \tilde\A_0] = 2c\tilde\A - c^2\A,   \label{eq:AA0}
\ee

\be
[\tilde A, \A_0] = 2c\A - c^2\tilde\A, \ \ \ \ \ \ [\tilde A, \tilde\A_0] = \tilde\A.   \label{eq:tAtA0}
\ee

\noindent Equation (\ref{eq:0}) is in fact just the definition of $F$ in terms of $T$. Plugging it into the right-hand sides of (\ref{eq:AG}) and (\ref{eq:tAG}) and applying the last relations allows one to integrate this first couple of linear PDE and obtain expressions for $G$ and $\tilde G$ in terms of $T$:

\be
G = -(\A_0 - 1)\frac{\tilde\A T}{2c}, \ \ \ \ \ \ \tilde G = -(\tilde\A_0 - 1)\frac{\A T}{2c}.  \label{eq:GG}
\ee

\ni The last expressions already allow to recover the Adler-van Moerbeke equation~\cite{AvM1, AvM3} in our approach. It follows in fact from a trivial identity:

$$
\A\tilde\A\ln\frac{U}{W} = \tilde\A\A\ln\frac{U}{W},
$$

\ni rewritten as

$$
\A\left(\frac{\tilde G}{F}\right) = \tilde\A\left(\frac{G}{F}\right).
$$

\ni Plugging in expressions for $F$, $G$ and $\tilde G$ in terms of $T$, we get a single PDE for $T\equiv \ln\tau_n^J$:

\begin{theorem}

The logarithm of the joint spacing probability for two coupled Gaussian matrices (or two time-point distribution for the Dyson process) satisfies the third-order PDE,

\be
\A\frac{(\tilde\A_0 - 1)\frac{\A T}{c}}{\A\tilde\A T + 2cn} = \tilde\A\frac{(\A_0 - 1)\frac{\tilde\A T}{c}}{\A\tilde\A T + 2cn},  \label{eq:AvM}
\ee

\end{theorem} 

\ni which is nothing but the equation first derived in~\cite{AvM1} using {\it two first} ($k = 1,2$) bilinear identities of the Toda series, while we started from two couples of first identities of AKNS (non-linear Schr\"odinger) series instead of the second ($k =2$) Toda identity. We, however, can derive also equations (\ref{eq:T1}) and (\ref{eq:T2}) of theorem \ref{thm:genT} from the system of theorem \ref{thm:1}, see Appendix A. It seems that in general these equations are hard to simplify.










\par Further on we will consider the case of only one boundary point for each matrix, i.e. study the joint largest eigenvalue distribution. Let us then rewrite everything in a different way, which will be convenient for the comparison with the variant of TW approach below. 




\ni Recalling the definitions of operators $\D_+$, $\D_-$, introduced in section \ref{sec:results}, we get 

$$
\A = \frac{(1+c)}{2}(\D_+ + \sigma\D_-), \ \ \ \tilde\A = \frac{(1+c)}{2}(\D_+ - \sigma\D_-),
$$

$$
\A^2 -\tilde\A^2 = (1-c^2)\D_+\D_-, \ \ \ \A^2 + \tilde\A^2 = \frac{(1+c)^2}{2}(\D_+^2 + \sigma^2\D_-^2),
$$

$$
\B_0 + \tilde\B_0 = 1/2(\xi_+\D_+ + \xi_-\D_-), \ \ \ \B_0 - \tilde\B_0 = 1/2(\xi_+\D_- + \xi_-\D_+).
$$

\noindent The system we obtained in theorem \ref{thm:1} now acquires the following form, if we take differences and sums of its last four equations:

$$
F = \frac{(\D_+^2 - \sigma^2\D_-^2) T}{4(1-\sigma^2)} + \frac{n}{2}, \eqno(\ref{eq:0+})  
$$

\be
\D_+\D_- U = - (\xi_+\D_- + \xi_-\D_+)U - 2\D_+\D_-T\cdot U,  \label{eq:(+t)-(+s)}
\ee 

\be
\D_+\D_- W = (\xi_+\D_- + \xi_-\D_+) W - 2\D_+\D_-T\cdot W,  \label{eq:(-t)-(-s)}
\ee 

\be
(\D_+^2 + \sigma^2\D_-^2) U = -(1+\sigma^2)(\xi_+\D_+ + \xi_-\D_-) U + 8\sigma c\prt_cU - 2(\D_+^2 + \sigma^2\D_-^2)T\cdot U,  \label{eq:(+t)+(+s)}
\ee 

\be
(\D_+^2 + \sigma^2\D_-^2) W = (1+\sigma^2)(\xi_+\D_+ + \xi_-\D_-) W - 8\sigma c\prt_cW - 2(\D_+^2 + \sigma^2\D_-^2)T\cdot W.  \label{eq:(-t)+(-s)}
\ee

\noindent Let 

\be
G_+ = W\D_+ U - U\D_+ W, \ \ \ G_- =  W\D_- U - U\D_- W, \ \ \ \hat G_c = Wc\prt_c U - Uc\prt_c W.   \label{eq:G+-}
\ee

\noindent Then from the first two equations for $U$, $W$ we derive one linear and one non-linear equation:

\be
\D_+(G_- + 2\xi_-F) + \D_-(G_+ + 2\xi_+F) = 0,   \label{eq:1}
\ee

\be
2F\D_+\D_-F - \D_+F\D_-F + (G_- + 2\xi_-F)(G_+ + 2\xi_+F) = 4(\xi_+\xi_- - 2\D_+\D_-T)F^2.  \label{eq:2}
\ee

\noindent From the longer pair of equations also one linear and one non-linear PDE result. They are

\be
\D_+(G_+ + 2\xi_+F) + \sigma^2\D_-(G_- + 2\xi_-F) = (1-\sigma^2)(\xi_+\D_+ + \xi_-\D_-)F + 8\sigma c\prt_cF + 4(1+\sigma^2)F,   \label{eq:3}
\ee

$$
2F (\D_+^2 + \sigma^2\D_-^2)F - (\D_+F)^2 - \sigma^2(\D_-F)^2 +   
$$

$$
+ (G_+ + 2\xi_+F)^2 + \sigma^2(G_- + 2\xi_-F)^2 - 2(1-\sigma^2)(\xi_+(G_+ + 2\xi_+F) - \xi_-(G_- + 2\xi_-F))F = 
$$

\be
= 16\sigma FG_c - 8(\D_+^2 + \sigma^2\D_-^2)\ln\tau_n^J F^2 + 4(\xi_+^2\sigma^2 + \xi_-^2)F^2.  \label{eq:4}
\ee

\noindent The pair of linear equations here is, of course, satisfied by the previously found general expressions for $G$ and $\tilde G$. Since now

\be
G = \frac{(1+c)}{2}(G_+ + \sigma G_-), \ \ \ \ \ \ \tilde G = \frac{(1+c)}{2}(G_+ - \sigma G_-),
\ee

\noindent they correspond to expressions \ref{eq:G+} and \ref{eq:G-} for $G_+$ and $G_-$, given in section \ref{sec:results} before theorem \ref{thm:F4T}.

 
 

\noindent The last expressions, together with (\ref{eq:0+}), are to be plugged into eqs.~(\ref{eq:2}) and (\ref{eq:4}). After that eq.~(\ref{eq:2}) becomes a single 4-th order PDE for $T\equiv\ln\tau_n^J$, that of theorem \ref{thm:F4T}.






\par As for the equation (\ref{eq:4}), the obstacle to obtaining another PDE for $T$ from it remains, since we have not found a convenient expression for $\hat G_c$. We can get much more from the matrix kernel approach~\cite{TW2}, which is the subject of the next sections. Nevertheless, the considerations presented here will be very helpful there to decipher the complicated system of matrix equations and to find more of the universal relations like those in~\cite{UniUE} for one-matrix UE.

\section{Matrix Kernel and TW equations for coupled Gaussian matrices}  \label{sec:4}

We consider sets $J_k = (\xi_k, \infty)$, i.e.~the simplest case of only largest eigenvalues and defer the due generalizations to several spectral gaps for a future work. Then

$$
\det(I - K^J) = \Pr(M(t_1) < \xi_1, \dots, M(t_m) < \xi_m).    
$$

\noindent We first follow~\cite{TW2} here and start with the general formula:

\be
\prt_k K^J =  - K\delta_k,   \label{eq:4.1}
\ee

\noindent where $\delta_k$ is the diagonal matrix with all entries zero except for the $k$-th and $(\delta_k)_{kk} = \delta(y-\xi_k)$. Then we introduce the (matrix) resolvent operator $R = K^J(I-K^J)^{-1}$ and see that

\be
\prt_k \ln\det(I-K^J) = -\text{Tr}(I-K^J)^{-1}\prt_kK^J = R_{kk}(\xi_k, \xi_k).   \label{eq:4.2}
\ee

\noindent We keep most of the notation in~\cite{TW2} i.e.~let $D = d/dx$, $M = x\cdot$ -- the multiplication operator, $\rho = (I-K^J)^{-1} = I + R$, $\xi = \text{diag}(\xi_k)$, $d\xi = \text{diag}(d\xi_k)$, $\chi(x) = \text{diag}(\chi_{J_k}(x))$, $\delta = \sum_k \delta_k$, $\Theta$ -- matrix with all elements equal unity; introduce matrices $r$, $r_x$, $r_y$ such that

\be
r_{ij} = R_{ij}(\xi_i, \xi_j), \ \ \ (r_x)_{ij} = (\prt_x R)_{ij}(\xi_i,\xi_j), \ \ \ (r_y)_{ij} = (\prt_y R)_{ij}(\xi_i,\xi_j).
\ee

\noindent Define also (here our normalization is as e.g.~in~\cite{TW1} for 1-matrix models rather than the one in~\cite{TW2})

\be
\varphi = b_{n-1}^{1/2}\varphi_n, \ \ \ \psi = b_{n-1}^{1/2}\varphi_{n-1},
\ee

\noindent where $b_{n-1} = \sqrt{n/2}$ for the Gaussian case, and introduce matrix functions

\be
Q = \rho\varphi, \ \ \ P = \rho\psi, \ \ \ \tilde Q = \varphi\chi\rho, \ \ \ \tilde P = \psi\chi\rho,   \label{eq:QPQP}
\ee

\noindent and their values at the spectral endpoints

\be
q_{ij} = Q_{ij}(\xi_i), \ \ \ \tilde q_{ij} = \tilde Q_{ij}(\xi_j), \ \ \ p_{ij} = P_{ij}(\xi_i), \ \ \ \tilde p_{ij} = \tilde P_{ij}(\xi_j).   \label{eq:qpqp}
\ee

\noindent Following~\cite{TW2}, also consider instead of (\ref{eq:00.4}) a modified kernel with the same Fredholm determinant:

\be
K_{ij}(x, y) = \left\{ \begin{array}{lr} \sum_{k=0}^{n-1}e^{(k-n)(t_i-t_j)}\varphi_k(x)\varphi_k(y) & \text{if }i \ge j,  \\ \ & \ \\  -\sum_{k=n}^{\infty}e^{(k-n)(t_i-t_j)}\varphi_k(x)\varphi_k(y) & \text{if }i < j. \end{array} \right.  \label{eq:4.3}
\ee 

\ni Using the well-known differentiation formulas for the harmonic oscillator eigenfunctions:

$$
\frac{d\varphi_k}{dx} = -x\varphi_k + \sqrt{2k}\varphi_{k-1}, \ \ \ \frac{d\varphi_{k-1}}{dx} = x\varphi_{k-1} - \sqrt{2k}\varphi_k,
$$

\noindent we have 

$$
(D+M)\varphi = \sqrt{2n}\psi, \ \ \ \ (D-M)\psi = -\sqrt{2n}\varphi.
$$

\noindent Then the two operator identities called Lemma 3 on p.24 of~\cite{TW2} read:

$$
(D+M)K_{ij} - e^{t_i-t_j}K_{ij}(D+M) = -2\psi(x)\varphi(y),
$$

$$
e^{t_i-t_j}(D-M)K_{ij} - K_{ij}(D-M) = -2\varphi(x)\psi(y).   
$$

\noindent The last formulas multiplied on the right by the matrix $\chi$, together with relation $[D, \chi] = \delta$ (and $[M, \chi] = 0$), if we let $t = \text{diag}(t_i)$ and so $e^t = \text{diag}(e^{t_i})$, lead to 

\be
e^{-t}(D+M)K^J - K^Je^{-t}(D+M) = -2e^{-t}\psi(x)\Theta\chi(y)\varphi(y) + K\delta e^{-t},   \label{eq:4.4}
\ee

\be
e^t(D-M)K^J - K^Je^t(D-M) = -2\varphi(x)\Theta\chi(y)\psi(y)e^t + K\delta e^t.    \label{eq:4.5}
\ee

\noindent Multiplying (\ref{eq:4.4}) and (\ref{eq:4.5}) by $\rho$ on the left and on the right, replacing $K^J \to K^J - I$ on the left-hand sides and using the above definitions we get the two equations of Lemma 4 of~\cite{TW2} (with the correction of swapping the matrices $e^t$ and $\Theta$ in the last formula below): 

\be
e^{-t}(D+M)R - Re^{-t}(D+M) = -2P(x)e^{-t}\Theta\tilde Q(y) + R\delta e^{-t}\rho, \label{eq:4.4'}
\ee

\be
e^t(D-M)R - Re^t(D-M) = -2Q(x)\Theta e^t\tilde P(y) + R\delta e^t\rho. \label{eq:4.5'}
\ee

\par At this point we depart from the line of~\cite{TW2}. Recall now that we consider the case of two matrices, in other words only two time points. Therefore, denoting $c = e^{-(t_2-t_1)}$ and introducing matrices

\be
e_L = \frac{(1+c)}{2}I + \frac{(1-c)}{2}\sigma_3,\ \ \ \  e_U = \frac{(1+c)}{2}I - \frac{(1-c)}{2}\sigma_3,
\ee

\noindent where $\sigma_3$ is the usual 3rd Pauli matrix, we can rewrite (\ref{eq:4.4'}) and (\ref{eq:4.5'}) as, respectively,

\be
e_L(D+M)R - Re_L(D+M) = -2P(x)e_L\Theta\tilde Q(y) + R\delta e_L\rho, \label{eq:4.6}
\ee

\be
e_U(D-M)R - Re_U(D-M) = -2Q(x)\Theta e_U\tilde P(y) + R\delta e_U\rho, \label{eq:4.7}
\ee

\ni Take $i, j$ entries in (\ref{eq:4.6}) and (\ref{eq:4.7}) and set $x = \xi_i$, $y = \xi_j$. This gives

\be
e_Lr_x + r_ye_L = -e_L\xi r + re_L\xi - 2pe_L\Theta\tilde q + re_Lr,     \label{eq:4.8}
\ee

\be
e_Ur_x + r_ye_U = e_U\xi r - re_U\xi - 2q\Theta e_U\tilde p + re_Ur,   \label{eq:4.9}
\ee

\noindent Since $e_L + e_U = (1+c)I$, $e_L - e_U = (1-c)\sigma_3$, adding and subtracting the equations (\ref{eq:4.8}), (\ref{eq:4.9}), we get, respectively,

\be
r_x + r_y - r^2 = -p(I+\sigma)\Theta\tilde q - q\Theta(I-\sigma)\tilde p - [\sigma\xi, r],   \label{eq:4.10} 
\ee

\be
\sigma r_x + r_y\sigma - r\sigma r = -p(I+\sigma)\Theta\tilde q + q\Theta(I-\sigma)\tilde p - [\xi, r],   \label{eq:4.11} 
\ee

\noindent where we have introduced the matrix $\sigma$,

\be
\sigma = \frac{1-c}{1+c}\sigma_3.
\ee

\par To derive our version of first order TW system, we introduce the matrix function $u$,

\be
u = (\varphi\chi, Q) = (\tilde Q, \varphi),     \label{eq:u} 
\ee

\noindent analogous to the scalar function $u$ for one-matrix case. It has not been used by Tracy and Widom in~\cite{TW2} but had already appeared in their earlier work~\cite{TW-AiryPr}, where the matrix Airy kernel for the Airy process was considered. Similarly, we introduce the matrix $w$, the analog of the scalar function $w$ from~\cite{TW1}:

\be
w = (\psi\chi, P) = (\tilde P, \psi)   \label{eq:w}. 
\ee

\noindent Acting by operator equation (\ref{eq:4.6}) upon function $\varphi$ from the left, we get

$$
\left[\frac{(I+\sigma)}{2}(D+M), \rho\right]\varphi = \frac{(I+\sigma)}{2}(D + x)Q - \sqrt{\frac{n}{2}}P(I+\sigma) = -P(I+\sigma)\Theta u + \frac{1}{2}R\delta(I+\sigma)Q,
$$

\noindent i.e.

\be
(I+\sigma)(DQ + xQ) = 2P(I+\sigma)\left(\sqrt{\frac{n}{2}} - \Theta u\right) + R\delta(I+\sigma)Q.   \label{eq:Q}
\ee

\noindent Similarly, acting by equation (\ref{eq:4.7}) upon function $\psi$ from the left gives



\be
(I-\sigma)(DP - xP) = -2Q\left(\sqrt{\frac{n}{2}} + \Theta(I-\sigma)w(I-\sigma)^{-1}\right)(I-\sigma) + R\delta(I-\sigma)P.   \label{eq:P}
\ee

\noindent On the other hand, acting by (\ref{eq:4.7}) upon $\varphi\chi$ and by (\ref{eq:4.6}) upon $\psi\chi$ {\it from the right} gives, respectively,



\be
(D\tilde Q + \tilde Qy)(I-\sigma) = 2\left(\sqrt{\frac{n}{2}} - u\Theta\right)(I-\sigma)\tilde P + \tilde Q(I-\sigma)\delta\cdot R,   \label{eq:tQ}
\ee

\noindent and



\be
(D\tilde P - \tilde Py)(I+\sigma) = -2(I+\sigma)\left(\sqrt{\frac{n}{2}} + (I+\sigma)^{-1}w(I+\sigma)\Theta\right)\tilde Q + \tilde P(I+\sigma)\delta\cdot R.   \label{eq:tP}
\ee

\par Next we use the fact following from equation (\ref{eq:4.1}):

\be
\prt_k \rho = \rho\prt_kK^J\rho = -R\delta_k\rho   \label{eq:4.14}
\ee

\noindent to derive the formulas~\cite{TW2}:

\be
\D_+ q_{ij} = \frac{dQ}{dx}(\xi_i) - (rq)_{ij}, \ \ \ \ \ \ \D_+\tilde q_{ij} = \frac{d\tilde Q}{dx}(\xi_i) - (\tilde qr)_{ij},   \label{eq:4.15}
\ee

\be
\D_+ p_{ij} = \frac{dP}{dx}(\xi_i) - (rp)_{ij}, \ \ \ \ \ \ \D_+\tilde p_{ij} = \frac{d\tilde P}{dx}(\xi_i) - (\tilde pr)_{ij}.    \label{eq:4.16}
\ee

\noindent Similarly, we get the corresponding expression for $\D_-$-derivatives:

\be
\D_- q_{ij} = \sigma_3\frac{dQ}{dx}(\xi_i) - (r\sigma_3q)_{ij}, \ \ \ \ \ \ \D_-\tilde q_{ij} = \frac{d\tilde Q}{dx}(\xi_i)\sigma_3 - (\tilde q\sigma_3r)_{ij},    \label{eq:4.17}
\ee

\be
\D_- p_{ij} = \sigma_3\frac{dP}{dx}(\xi_i) - (r\sigma_3p)_{ij}, \ \ \ \ \ \ \D_-\tilde p_{ij} = \frac{d\tilde P}{dx}(\xi_i) - (\tilde p\sigma_3r)_{ij}.    \label{eq:4.18}
\ee

\noindent We take $x = \xi_i$, $y = \xi_j$ in the eqs.~(\ref{eq:Q}), (\ref{eq:P}), (\ref{eq:tQ}) and (\ref{eq:tP}) and use the previous formulas to obtain PDE with spectral endpoints $\xi_1, \xi_2$ as indepenent variables:

\be
(I+\sigma)\D_+q = -(I+\sigma)\xi q + 2p(I+\sigma)\left(\sqrt{\frac{n}{2}} - \Theta u\right) - [\sigma, r]q,   \label{eq:q+}
\ee

\be
(I-\sigma)\D_+p = (I-\sigma)\xi p - 2q\left(\sqrt{\frac{n}{2}} + \Theta(I-\sigma)w(I-\sigma)^{-1}\right)(I-\sigma) + [\sigma, r]p,   \label{eq:p+}
\ee

\be
\D_+\tilde q(I-\sigma) = -\tilde q\xi(I-\sigma) + 2\left(\sqrt{\frac{n}{2}} - u\Theta\right)(I-\sigma)\tilde p - \tilde q[\sigma, r],   \label{eq:tq+}
\ee

\be
\D_+\tilde p(I+\sigma) = \tilde p\xi(I+\sigma) - 2(I+\sigma)\left(\sqrt{\frac{n}{2}} + (I+\sigma)^{-1}w(I+\sigma)\Theta\right)\tilde p + \tilde p[\sigma, r].   \label{eq:tp+}
\ee

\noindent and the other four equations for $\D_-$ derivatives:

$$
(I+\sigma)\D_-q = -(I+\sigma)\sigma_3\xi q + 2\sigma_3p(I+\sigma)\left(\sqrt{\frac{n}{2}} - \Theta u\right) + \sigma_3r(I+\sigma)q - (I+\sigma)r\sigma_3q,   
$$

\noindent i.e.

\be
(I+\sigma)\D_-q = -(I+\sigma)\sigma_3\xi q + 2\sigma_3p(I+\sigma)\left(\sqrt{\frac{n}{2}} - \Theta u\right) + [\sigma_3, r]q,   \label{eq:q-}
\ee

\noindent and, similarly,

\be
(I-\sigma)\D_-p = (I-\sigma)\sigma_3\xi p - 2\sigma_3q\left(\sqrt{\frac{n}{2}} + \Theta(I-\sigma)w(I-\sigma)^{-1}\right)(I-\sigma) + [\sigma_3, r]p,   \label{eq:p-}
\ee

\be
\D_-\tilde q(I-\sigma) = -\tilde q\sigma_3\xi(I-\sigma) + 2\left(\sqrt{\frac{n}{2}} - u\Theta\right)(I-\sigma)\tilde p\sigma_3 - \tilde q[\sigma_3, r],   \label{eq:tq-}
\ee

\be
\D_-\tilde p(I+\sigma) = \tilde p\sigma_3\xi(I+\sigma) - 2(I+\sigma)\left(\sqrt{\frac{n}{2}} + (I+\sigma)^{-1}w(I+\sigma)\Theta\right)\tilde p\sigma_3 - \tilde p[\sigma_3, r].   \label{eq:tp-}
\ee

\ni Equations (\ref{eq:4.10}), (\ref{eq:4.11}) for the derivatives of $r$ together with a consequence of (\ref{eq:4.14}) (since $\prt_k R = \prt_k\rho$)~\cite{TW2},

$$
\prt_k r_{ij} = \prt_k(R_{ij}(\xi, \xi_j)) = (\prt_k R_{ij})(\xi_i, \xi_j) + \prt_xR_{ij}(\xi_i, \xi_j)\delta_{ik} + \prt_yR_{ij}(\xi_i, \xi_j)\delta_{jk} =
$$

\be
= -r_{ik}r_{kj} + \prt_xR_{ij}(\xi_i, \xi_j)\delta_{ik} + \prt_yR_{ij}(\xi_i, \xi_j)\delta_{jk},
\ee

\noindent give (we will sometimes write $\bar\D_-$ for $(\sigma_3\sigma)\D_-$)

\be
\D_+r = r_x + r_y - r^2 = -p(I+\sigma)\Theta\tilde q - q\Theta(I-\sigma)\tilde p - [\sigma\xi, r],   \label{eq:r+} 
\ee

\be
\bar\D_-r = \sigma r_x + r_y\sigma - r\sigma r = -p(I+\sigma)\Theta\tilde q + q\Theta(I-\sigma)\tilde p - [\xi, r].   \label{eq:r-} 
\ee

\noindent Besides, differentiating the definitions (\ref{eq:u}), (\ref{eq:w}) and using (\ref{eq:QPQP}), (\ref{eq:qpqp}) and (\ref{eq:4.1}), we get the matrix analogs of universal equations~\cite{TW1} for scalar $u$ and $w$:

$$
\D_+u = -\tilde qq, \ \ \  \D_-u = -\tilde q\sigma_3q, \ \ \ \D_+w = -\tilde pp, \ \ \ \D_-w = -\tilde p\sigma_3p.
$$

Thus, we have obtained a system of fourteen first-order matrix PDE for joint largest eigenvalue distribution of two coupled Gaussian matrices. Our system is different from that in~\cite{TW2} since we employ the matrices $u$ and $w$. As we will see, this system is again the most convenient for comparison with $\tau$-function approach of Adler and van Moerbeke~\cite{AvM1, AvM3} like it was the case for one-matrix ensembles~\cite{IR1, UniUE}. 
\par We remark that we are interested in the matrix $r$ while equations (\ref{eq:r+}), (\ref{eq:r-}) involve only terms where the other (auxiliary) variables enter multiplied by the degenerate constant matrix $\Theta$.  Therefore, since $\Theta^2 = 2\Theta$ and

\be
(I - \sigma)(I + \sigma) = I - \sigma^2 = \left(1 - \frac{(1-c)^2}{(1+c)^2}\right)\cdot I = \frac{4c}{(1+c)^2}\cdot I,
\ee

\noindent it is convenient to introduce new matrix variables, which in fact amounts to some reduction of the total number of scalar variables since the new matrices, being proportional to $\Theta$, also have determinant zero:

\be
\hat q = q\Theta, \ \ \ \hat{\tilde q} = \Theta\tilde q, \ \ \ \hat p = \frac{(I-\sigma)p(I+\sigma)}{1-\sigma^2}\Theta, \ \ \ \hat{\tilde p} = \Theta\frac{(I-\sigma)\hat{\tilde p}(I+\sigma)}{1-\sigma^2},
\ee

\be
\hat U \equiv \hat U_n = \sqrt{\frac{n}{2}}\Theta - \Theta u_n \Theta, \ \ \ \hat W \equiv \hat W_n = \sqrt{\frac{n}{2}}\Theta + \Theta \frac{(I-\sigma)w_n(I+\sigma)}{1-\sigma^2} \Theta.
\ee

\noindent Here and further on, in a slight abuse of notation, we write $\sigma^2$ for the scalar $(1-c)^2/(1+c)^2$, so $1-\sigma^2$ stands for $4c/(1+c)^2$. To express everything in terms of these new variables, we multiply equations (\ref{eq:q+}), (\ref{eq:p+}), (\ref{eq:q-}), (\ref{eq:p-}) by the matrix $\Theta$ on the right and equations (\ref{eq:tq+}), (\ref{eq:tp+}), (\ref{eq:tq-}), (\ref{eq:tp-}) by the matrix $\Theta$ on the left. This brings our TW-type system to the form given in theorem \ref{thm:QPsys}, thus finishing its proof.








  









It is easy to verify that the system has two matrix first integrals, similar to the first integral of~\cite{TW1} for Gaussian single matrices:

\be
\hat{\tilde p}\hat q = n\Theta - \hat W\hat U,  \label{eq:WU}
\ee

\be
\hat{\tilde q}\hat p = n\Theta - \hat U\hat W.  \label{eq:UW}
\ee

\section{Coupled analogs of one-matrix relations among matrix kernel related variables and $\tau$-functions}   \label{sec:unirel}

There are direct analogs of one-matrix universal relations~\cite{UniUE} here.

\begin{lemma}

Analogs of rank-1 projection operator kernels read:

\be
K_{n+1}e_L - e_LK_n = e_L\Theta\varphi_n(x)\varphi_n(y)\chi(y),      \label{eq:Kl}
\ee

\be
e_UK_{n+1} - K_ne_U = \Theta e_U\varphi_n(x)\varphi_n(y)\chi(y).    \label{eq:Ku}
\ee

\end{lemma}

\begin{proof}

The diagonal entries of the matrix kernel $K$ are the same as in scalar (1-matrix) case. Consider the anti-diagonal entries of the kernels for matrices of consecutive sizes $n$ and $n+1$. Since 

$$
(K_n)_{21}(x, y) = \sum_{k=0}^{n-1}c^{n-k}\varphi_k(x)\varphi_k(y)\chi(y), \ \ \ (K_n)_{12}(x, y) = -\sum_{k=n}^{\infty}c^{k-n}\varphi_k(x)\varphi_k(y)\chi(y),
$$

\ni we find

$$
\frac{1}{c}(K_{n+1})_{21} - (K_n)_{21}(x, y) = \varphi_n(x)\varphi_n(y)\chi(y), \ \ \ c(K_{n+1})_{12} - (K_n)_{12}(x, y) = \varphi_n(x)\varphi_n(y)\chi(y),
$$

\ni which means, that we have matrix equation

$$
\frac{1}{c}e_UK_{n+1}(x,y)e_L - K_n(x,y) = \Theta\varphi_n(x)\varphi_n(y)\chi_J(y).
$$

\ni By the identity $e_Le_U = e_Ue_L = cI$, it is equivalent to the statement of the lemma.

\end{proof}

\begin{lemma}

There are the corresponding relations for the resolvent kernels:

\be
R_{n+1}e_L - e_LR_n = P_{n+1}(x)e_L\Theta\tilde Q_n(y),     \label{eq:Rl}
\ee

\be
e_UR_{n+1} - R_ne_U = Q_n(x)\Theta e_U\tilde P_{n+1}(y).     \label{eq:Ru}
\ee

\end{lemma}

\begin{proof}

Writing $K_{n+1}e_L - e_LK_n = e_L(I - K_n) - (I - K_{n+1})e_L$ and multiplying (\ref{eq:Kl})  by $I + R_n = (I - K_n)^{-1}$ on the right and by $I + R_{n+1} = (I - K_{n+1})^{-1}$ on the left gives equation (\ref{eq:Rl}). Equation (\ref{eq:Ru}) follows from (\ref{eq:Ku}) similarly.

\end{proof}
 
\noindent Also by definition of matrices $Q_n$, $P_n$, $\tilde Q_n$ and $\tilde P_n$ we have recursion relations:

\be
\varphi_n = (I-K_n)Q_n = (I-K_{n+1})P_{n+1} = \tilde Q_n(I-K_n) = \tilde P_{n+1}(I-K_{n+1}).    \label{eq:QP}
\ee

\noindent They in turn lead to relations involving inner product matrices $\bar u_n = (\varphi_n, (I-K_n)\varphi_n)$ and $\bar w_n = (\varphi_{n-1}, (I-K_n)\varphi_{n-1})$ :

\begin{lemma}

$$
e_UP_n = Q_{n-1}e_U + Q_{n-1}\Theta e_U\bar w_n, \ \ \ \tilde P_{n+1}e_L = e_L\tilde Q_n + \bar w_{n+1}e_L\Theta\tilde Q_n,
$$

$$
e_LQ_n = P_{n+1}e_L - P_{n+1}e_L\Theta \bar u_n, \ \ \ \tilde Q_{n-1}e_U = e_U\tilde P_n - \bar u_{n-1}\Theta e_U\tilde P_n.
$$

\end{lemma}

\begin{proof}

Let us prove e.g.~the first of the above formulas: we have

$$
\varphi_ne_U = e_U(I-K_n)Q_n = e_U(I-K_{n+1})P_{n+1} = (I-K_n)e_UP_{n+1} - (e_UK_{n+1} - K_ne_U)P_{n+1} = 
$$

$$
= (I-K_n)e_UP_{n+1} - \varphi_n\Theta e_U\bar w_{n+1},
$$

\ni and acting on both sides by $(I-K_n)^{-1}$ from the left gives

$$
e_UP_{n+1} =  Q_ne_U + Q_n\Theta e_U\bar w_{n+1},
$$

\ni i.e. the sought formula if we shift $n \to n-1$. The other three formulas are obtained quite similarly. 

\end{proof}

Introduce matrix $\hat{\bar w}_n = e_U\bar w_ne_L/c$, then we have

\begin{corollary}

\be
I + \Theta\hat{\bar w}_{n+1} = (I - \Theta \bar u_n)^{-1}
\ee

\noindent and

\be
I + \hat{\bar w}_{n+1}\Theta = (I - \bar u_n\Theta)^{-1}.
\ee

\end{corollary}

\begin{proof}

One can rewrite recursion relations from previous lemma as follows:    

$$
\frac{e_UP_{n+1}e_L}{c} = Q_n(I + \Theta\hat{\bar w}_{n+1}) = Q_n(I - \Theta \bar u_n)^{-1}, 
$$

$$
\frac{e_U\tilde P_{n+1}e_L}{c} = (I + \hat{\bar w}_{n+1}\Theta)\tilde Q_n = (I - \bar u_n\Theta)^{-1}\tilde Q_n,
$$



\ni thus we get the statement.



\end{proof}

\ni It follows also for the matrix resolvent kernels that

$$
R_{n+1}e_L - e_LR_n = e_LQ_n(x)(I - \Theta \bar u_n)^{-1}\Theta\tilde Q_n(y),     
$$

$$
e_UR_{n+1} - R_ne_U = Q_n(x)\Theta (I - \bar u_n\Theta)^{-1}\tilde Q_n(y)e_U.     
$$

\ni Taking $x = \xi_i$, $y = \xi_j$ in the last formulas gives, respectively,

$$
r_{n+1}e_L - e_Lr_n = e_Lq_n(I - \Theta \bar u_n)^{-1}\Theta\tilde q_n,     
$$

$$
e_Ur_{n+1} - r_ne_U = q_n\Theta (I - \bar u_n\Theta)^{-1}\tilde q_ne_U.     
$$

\ni Now, after recalling again $e_Le_U = e_Ue_L = cI$, it follows from either of these two formulas that

$$
\D_+ \ln\frac{\tau_{n+1}^J}{\tau_n^J} = \text{Tr} (r_{n+1} - r_n) = \text{Tr}\left((I - \Theta \bar u_n)^{-1}\D_+(I - \Theta \bar u_n)\right) = \D_+\ln\det(I - \Theta \bar u_n),
$$

\ni which can be integrated to give

$$
\frac{\tau_{n+1}^J/\tau_{n+1}}{\tau_n^J/\tau_n} = \det(I - \Theta \bar u_n) = \det(I - \bar u_n\Theta) = 1 - \text{Tr}(\Theta \bar u_n)
$$

\ni Using the recursion relations to get the second, similar formula, we thus prove the main result of this section -- theorem \ref{thm:unirel}, since we also get the formulas for the matrices $\hat{\bar U}$ and $\hat{\bar W}$ defined before it in section \ref{sec:results}:

$$
\text{Tr}\hat U = \sqrt{\frac{n}{2}}\text{Tr}\hat{\bar U} = \sqrt{2n} - 2\text{Tr}(\Theta u_n) = \sqrt{2n}(1 - \text{Tr}(\Theta \bar u_n)) = \sqrt{2n}\;\frac{\tau_{n+1}^J/\tau_{n+1}}{\tau_n^J/\tau_n},
$$


$$
\text{Tr}\hat W = \sqrt{\frac{n}{2}}\text{Tr}\hat{\bar W} = \sqrt{2n} + 2\text{Tr}(\Theta \hat w_n) = \sqrt{2n}(1 + \text{Tr}(\Theta \hat{\bar w}_n)) = \sqrt{2n}\;\frac{\tau_{n-1}^J/\tau_{n+1}}{\tau_n^J/\tau_n}.
$$


 





\ni In fact, one has a remarkable result: \\

\bigskip
{\bf\large Lemma \ref{lem:tr}.} 


\be
\text{Tr}(\hat U\hat W) = \text{Tr}\hat U\cdot \text{Tr}\hat W
\ee

\bigskip 

\begin{proof}

Since Tr$r = \D_+\ln\tau_n^J$, Tr$(\sigma_3r) = \D_-\ln\tau_n^J$, Tr$(\sigma r) = \bar\D_-\ln\tau_n^J$, we get from the equations (\ref{eq:r+}) and (\ref{eq:r-}):

$$
(\D_+^2 - \bar\D_-^2)\ln\tau_n^J = -\frac{(1-\sigma^2)}{2}\text{Tr}(\hat q\hat{\tilde p} + \hat p\hat{\tilde q}).
$$

\ni We use the first integrals (\ref{eq:WU}) and (\ref{eq:UW}) to find

$$
\text{Tr}(\hat q\hat{\tilde p}) = \text{Tr}(\hat p\hat{\tilde q}) = \text{Tr}(n\Theta - \hat U\hat W) = 2n - \text{Tr}(\hat U\hat W),
$$

\ni so the previous equation can be written as

$$
(\D_+^2 - \bar\D_-^2)\ln\tau_n^J = 2(1-\sigma^2)\left(\frac{1}{2}\text{Tr}(\hat U\hat W) - n\right).
$$

\ni We compare this last equation with the ``boundary-Toda" equation from section 3, which can be written as

$$
(\D_+^2 - \bar\D_-^2)\ln\tau_n^J = 2(1-\sigma^2)(2(1-c^2)U_nW_n - n),
$$

which immediately leads to

$$
\text{Tr}(\hat U\hat W) = 4(1-c^2)U_nW_n.
$$

\ni Therefore, by the previous theorem, see also section 2 for the normalization of the 2-matrix integrals $\tau_n$ over the whole domain,





$$
\text{Tr}\hat U \cdot \text{Tr}\hat W = 2nU_nW_n\;\frac{\tau_n^2}{\tau_{n+1}\tau_{n-1}} = 4(1-c^2)U_nW_n.
$$

\end{proof}
  

\section{Transformation of the TW-type system and PDE analogs of Painlev\'e IV}   \label{sec:Xsys}

We consider and solve the system obtained in theorem \ref{thm:Xsys}.

\begin{lemma}

The system of theorem \ref{thm:Xsys} has three full matrix first integrals: 

\be
[X_+, G] = \{X_-, \T\},   \label{eq:I1}
\ee

\be
[X_+, \T] = \{X_-, G\},    \label{eq:I2}
\ee

\be
[\T, G] = \{X_-, 3X_+^2 - 8nX_+ + X_-^2\} - [X_+, [X_+, X_-]],    \label{eq:I3}
\ee

\end{lemma}

\begin{proof}

The first two formulas are just the integrals (\ref{eq:WU}) and (\ref{eq:UW}) rewritten in new variables, and the third one can be easily derived from the above system of PDE as consistency condition, if using the first two. Consider e.g.

$$
[\D_+X_+, G] + [X_+, \D_+G] = [\T, G] - [[\xi+\tilde A, X_-], G] + [X_+, \{\xi+\tilde A, \T\}] + [X_+, [X_+, X_-]],
$$

\ni i.e.

$$
[\T, G] = \D_+[X_+, G] - [X_+, [X_+, X_-]] + [[\xi+\tilde A, X_-], G] - [X_+, \{\xi+\tilde A, \T\}] = 
$$

$$
= \D_+\{X_-, \T\} - [X_+, [X_+, X_-]] + [[\xi+\tilde A, X_-], G] - [X_+, \{\xi+\tilde A, \T\}] = 
$$

$$
= \{X_-, 3X_+^2 - 8nX_+ + X_-^2\} - [X_+, [X_+, X_-]] + 
$$

$$
+ [[\xi+\tilde A, X_-], G] - [X_+, \{\xi+\tilde A, \T\}] - \{[\xi+\tilde A, X_+], \T\} + \{X_-, \{\xi+\tilde A, G\}\},
$$

\ni which gives (\ref{eq:I3}), since the last four terms on the last line cancel out due to (\ref{eq:I2}).

\end{proof}

\noindent There are additional {\it diagonal matrix} first integrals, which can be obtained by integrating the diagonal parts of equations (\ref{eq:D+X-}) and, using expression for $\T$ from eq.~(\ref{eq:D+X+}), also (\ref{eq:D+G}), or diagonal parts of equations (\ref{eq:D-X-}) and (\ref{eq:D-G}), using expression for anti-diagonal part of $G$ from eq.~(\ref{eq:D-X-}) on the right-hand side of (\ref{eq:D-G}). The pairs of equations obtained in these two ways become identical after integration. Thus we get, respectively,

\begin{lemma}

The system has two diagonal matrix first integrals,

\be
(X_-)_d = -A^2\sigma,   \label{eq:X-d}
\ee

\be
G_d = 4r_d - 2\xi_+\D r_d - 2\xi_-\D_- r_d + \{\tilde A, (X_+)_a\} - \sigma_3\{A, X_a\}.     \label{eq:Gd}
\ee

\end{lemma}

\bigskip
\par Splitting into diagonal/anti-diagonal parts appears to be convenient for the rest of the system also. Equations for the derivatives of $r$ in diagonal/anti-diagonal splitting, if we denote 


$$
X_d \equiv (X_+)_d, \ \ \ X_a \equiv (X_-)_a,
$$

\noindent are

\be
2\D_+ r_d = - X_d + \sigma(X_-)_d, \ \ \ \ 2\bar\D_-r_d = -\sigma X_d + (X_-)_d,      \label{eq:rd}
\ee

\be
\sigma_3\D_+ A = - (X_+)_a - \xi_+\tilde A, \ \ \ \ \sigma\D_-A = X_a - \xi_-A.      \label{eq:ra} 
\ee 

\noindent The equation (\ref{eq:X-d}) then entails important simple relations

\be
X_d = -2\D_+ r_d - \sigma^2A^2 = -2\sigma_3\D_-r_d - A^2.     \label{eq:XrA}
\ee

\noindent Further splitting of the diagonal parts into scalar trace, we denote Tr$M$ by $M_t$, and tr$(\sigma_3M)$, denoted by $M_3$, parts will also be used. As it is clear from the eq.~(\ref{eq:D+X-}), Tr$X_- = 0$. As follows from the connection of the matrix $r$ with $\ln\tau_n^J$,

$$
\D_+\D_-\ln\tau_n^J = \D_+r_3 = \D_-r_t,
$$

\noindent (the second equality above can be seen also from each of the equations (\ref{eq:rd})). Also, a consequence of (\ref{eq:XrA}) is the formula \ref{eq:A^2} simply relating anti-diagonal elements of $r$ with its diagonal elements. 


\noindent One can see that, besides equations (\ref{eq:r+}), (\ref{eq:r-}), some combinations of the other equations in the system are also linear in $X_+, X_-, \T$ and $G$. Namely, adding (\ref{eq:D+X+}) and $(\sigma_3\sigma)\cdot$(\ref{eq:D-X-}) gives

$$
\D_+X_+ + (\sigma_3\sigma)\D_-X_- = \T - \xi_-\frac{[\sigma_3, X_-]}{2} - \xi_+\frac{[\sigma, X_+]}{2} + \frac{[\sigma, G]}{2},
$$

\noindent while adding (\ref{eq:D+X-}) and $(\sigma_3\sigma)\cdot$(\ref{eq:D-X+}) results in

$$
\D_+X_- + (\sigma_3\sigma)\D_-X_+ = - \xi_-\frac{[\sigma_3, X_+]}{2} - \xi_+\frac{[\sigma, X_-]}{2} + \frac{\{\sigma, \T\}}{2}.
$$

\noindent It follows from the linear equations that

\be
\T_t = \D_+ X_t = \D_-X_3 - 2\D_+ A^2,    \label{eq:TtX}
\ee

\be
\T_3 = \D_-X_t = \D_+ X_3 - 2\sigma^2\D_-A^2,   \label{eq:T3X}
\ee

\noindent recall that $\sigma^2 = (1-c)^2/(1+c)^2$. Recall the quantities $A_+$, $A_-$, $X_3$, $G_t$, $G_3$ defined in the corollary \ref{cor:main}, in section \ref{sec:results}. 


The splitting of eq.~(\ref{eq:Gd}) into the scalar parts reads:

\be
G_t = 4r_t - 2\xi_+\D_+r_t + \xi_-X_3 + A_+ \equiv H_t + A_+,   \label{eq:Gt}
\ee

\be
G_3 = 4r_3 - 2\xi_-\D_- r_3 + \xi_+X_3 + A_- \equiv H_3 + A_-.    \label{eq:G3}
\ee



\par In fact, one can show that

\begin{theorem}       \label{theorem:4ord}

The system can be reduced to the following five independent (scalar) 4th-order PDE in $\ln\tau_n^J$:

\be
X_3\D_+\D_-X_3 - \D_+X_3\D_-X_3 + R_tR_3 - (X_3 + \xi_+\xi_-)X_3^2 = 0,   \label{eq:S}
\ee

\noindent where

$$
R_t = 4r_t - 2\xi_+\D_+r_t, \ \ \ \ \ \ R_3 = 4r_3 - 2\xi_-\D_-r_3,   
$$

\be
\D_+\D_-X_t = \xi_-G_t + \xi_+G_3 + X_3(3X_t - 8n),   \label{eq:Tt3}
\ee

\be
2\hat FX_3\D_+^2X_t = 2\hat F(\D_+X_3\D_+X_t - R_tG_3 + \xi_+X_3G_t) + X_3(\D_+X_t^2 - G_t^2),    \label{eq:+Tt}
\ee

\be
2\hat FX_3\D_-^2X_t = 2\hat F(\D_-X_3\D_-X_t - R_3G_t + \xi_-X_3G_3) + X_3(\D_-X_t^2 - G_3^2)   \label{eq:-T3}
\ee

\ni and

\be
\frac{\sigma^2(\D_+^2-\D_-^2)X_3}{(1-\sigma^2)} = -2\sigma^2\D_+\D_-A^2 = -(6\sigma^2X_3A^2 - \xi_-A_+ - \sigma^2\xi_+A_-).    \label{eq:DDA}
\ee

\end{theorem}

\ni For the proof, see Appendix B and formulas (\ref{eq:+B}) and (\ref{eq:-B}) below.

\par 
It is rather tricky to integrate these equations directly. There is, however, a simpler way. Additional integrals can be most readily seen from the original defining variables:

$$
\det(\hat q) = \det(\hat{\tilde q}) = \det(\hat p) = \det(\hat{\tilde p}) = \det(\hat U) = \det(\hat W) = 0,
$$

\noindent so

\be
\det(X_+ \pm X_-) = \det(\T \pm G) = 0.   \label{eq:det}
\ee

\noindent There are also remarkable identities here:

\begin{lemma}

\be
Tr\left((\hat U\hat W)^k\right) = Tr\left((\hat W\hat U)^k\right).   \label{eq:trUW}
\ee

\end{lemma}

\begin{proof}

It follows by induction from the case $k=1$ proved in the previous section by comparison with the results from Toda lattice.

\end{proof}

\noindent Then we retrieve

\begin{lemma}

It follows from eq.~(\ref{eq:det}) that (besides eq.~(\ref{eq:XXa}) in Appendix B which arises this way also)

\be
4((X_+)_a^2 + X_a^2) = J \equiv X_t^2 - X_3^2 - 4\sigma^2(A^2)^2,    \label{eq:B+}
\ee

\be
4(\T_a^2 + G_a^2) = \T_t^2 - \T_3^2 + G_t^2 - G_3^2,    \label{eq:TG2+}
\ee

\be
2\{\T_a, G_a\} = G_t\T_t - G_3\T_3,   \label{eq:TGaa}
\ee

\noindent and from (\ref{eq:trUW}) one can get once again both eqs.~(\ref{eq:XXa}) and (\ref{eq:B+}), but also another first integral:

\be
4(G_a^2 - \T_a^2) = \T_t^2 + \T_3^2 - G_t^2 - G_3^2 - 4X_t^2(X_t - 4n).   \label{eq:TG2-}   
\ee

\end{lemma}

\noindent We will need three more auxiliary formulas, which are direct consequences of previously found first integrals (\ref{eq:I1a}), (\ref{eq:I2a}) in Appendix B:

\be
X_3^2G_a^2 = G_3^2(X_+)_a^2 - \T_t^2X_a^2 - \T_tG_3C_a,   \label{eq:Ga2}
\ee

\be
X_3^2\T_a^2 = \T_3^2(X_+)_a^2 - G_t^2X_a^2 - \T_3G_tC_a,   \label{eq:Ta2}
\ee

\be
X_3^2\{\T_a, G_a\} = 2(\T_3G_3(X_+)_a^2 - \T_tG_tX_a^2) - (\T_t\T_3 + G_tG_3)C_a.    \label{eq:TG}
\ee

\noindent First two of the last equations are obtained by taking square of eqs.~(\ref{eq:I1a}) and (\ref{eq:I2a}), respectively, and the third is their anti-commutator. Plugging (\ref{eq:B+}) into eq.~(\ref{eq:Ix}) of Appendix B simplifies it, giving

\be
\T_t\T_3 - G_tG_3 = 2X_3(X_t^2 - 4nX_t),   \label{eq:x}
\ee

\par Using these relations together with lemma 8, one can verify that all the higher-order equations are satisfied, see Appendix C. \\
\par One can easily eliminate all auxiliary variables but two: $A_+$, $A_-$, and get 

\begin{theorem}

The joint largest eigenvalue distribution for two coupled Gaussian matrices satisfies the system of PDE:

$$
\T_t\T_3 - G_tG_3 - 2X_tX_3\hat F = 0, \eqno(\ref{eq:x})
$$

\be
A_+A_- = 4\sigma^2(\D_+A^2\D_-A^2 - 2X_3(A^2)^2), \label{eq:Ax}
\ee

\be
(\T_t)^2 + (\T_3)^2 - G_t^2 - G_3^2 - 2\hat F(2X_3^2 + J) = 0, \label{eq:2+}
\ee

\be
A_+^2 + \sigma^2A_-^2 = 4\sigma^2(D_+^2 + D_-^2 + A^2J), \label{eq:A2+}
\ee

\be
\hat F(A_+^2 - 4\sigma^2(D_+^2 - X_3^2A^2)) = 2\sigma^2A^2(\T_t^2 - G_t^2), \label{eq:+}
\ee

\be
\hat F(A_-^2 - 4(\sigma^2D_-^2 - X_3^2A^2)) = 2A^2(\T_3^2 - G_3^2), \label{eq:-}
\ee

\be
\hat F(D_+A_- - D_-A_+) = A^2(\T_tG_3 - \T_3G_t), \label{eq:a}
\ee

\noindent with notations introduced:

$$
\hat F = X_t - 4n,  \ \ \ J = X_t^2 - X_3^2 - 4\sigma^2(A^2)^2, \ \ \ D_+ = \D_+ A^2, \ \ \ D_- = \D_-A^2.
$$

\noindent For the convenience of the reader we write out again the definitions of the variables entering the system above:

$$
G_t = H_t + A_+, \ \ \ G_3 = H_3 + A_-, \ \ \ H_t = 4r_t - 2\xi_+\D_+ r_t + \xi_-X_3,  \ \ \ H_3 = 4r_3 - 2\xi_-\D_- r_3 + \xi_+X_3,
$$

$$
\T_t = \D_+ X_t = \D_-X_3 - 2\D_+ A^2, \ \ \ \T_3 = \D_-X_t = \D_+ X_3 - 2\sigma^2\D_-A^2, \ \ \ \sigma^2 = (1-c)^2/(1+c)^2,
$$

$$
X_3 = -2\D_+ r_3 = -2\D_-r_t = -2\D\D_-\ln\tau_n^J.
$$

\end{theorem}

\begin{proof} 
We already have equations (\ref{eq:x}) and (\ref{eq:Ax}) (for the last see lemma 11 in Appendix B), equation (\ref{eq:A2+}) is the result of taking the combination of equations $(\ref{eq:A+}) + \sigma^2(\ref{eq:A-})$ from Appendix B and using (\ref{eq:B+}) to eliminate $(X_+)_a^2 + X_a^2$. After expressing $G_a^2$ and $\T_a^2$ from equations (\ref{eq:TG2+}) and (\ref{eq:TG2-}):

\be
4G_a^2 = \T_t^2 - G_3^2 - 2X_t^2\hat F,   \label{eq:E+}
\ee

\be
4\T_a^2 = -\T_3^2 + G_t^2 + 2X_t^2\hat F,   \label{eq:E-}
\ee

\noindent and plugging them into (\ref{eq:Ga2}) and (\ref{eq:Ta2}), respectively, the last become:

$$
G_3^2(4(X_+)_a^2 + X_3^2) - \T_t^2(4X_a^2 + X_3^2) = 4\T_tG_3C_a - 2X_3^2X_t^2\hat F,   \eqno(\ref{eq:Ga2})
$$

$$
\T_3^2(4(X_+)_a^2 + X_3^2) - G_t^2(4X_a^2 + X_3^2) = 4\T_3G_tC_a + 2X_3^2X_t^2\hat F.   \eqno(\ref{eq:Ta2})
$$

\noindent Applying (\ref{eq:I1a}), (\ref{eq:I2a}) and (\ref{eq:x}) transforms equation (\ref{eq:TG}) into



$$
\T_3G_3(4(X_+)_a^2 + X_3^2) - \T_tG_t(4X_a^2 + X_3^2) = 4(\T_t\T_3 - X_3X_t(X_t - 4n))C_a.   \eqno(\ref{eq:TG})
$$

\noindent Then we make combinations, $\T_t\cdot(\ref{eq:TG}) - G_t\cdot(\ref{eq:Ga2})$, which gives, after using (\ref{eq:x}),

$$
2X_3X_t\hat FG_3(4(X_+)_a^2 + X_3^2) = 4X_3X_t\hat F\T_tC_a + 2X_3^2X_t^2\hat FG_t,
$$

\noindent and $-\T_3\cdot(\ref{eq:TG}) + G_3\cdot(\ref{eq:Ta2})$, giving, after (\ref{eq:x}) is applied,

$$
2X_3X_t\hat FG_t(4X_a^2 + X_3^2) = 2X_3^2X_t^2\hat FG_3 - 4X_3X_t\hat F\T_3C_a.
$$

\noindent They are obviously simplified into

\be
G_3(4(X_+)_a^2 + X_3^2) = 2\T_tC_a + X_3X_tG_t,    \label{eq:GX+a}
\ee

\be
G_t(4X_a^2 + X_3^2) = X_3X_tG_3 - 2\T_3C_a.    \label{eq:GX-a}
\ee
 
\noindent Their combination $\T_3\cdot(\ref{eq:GX+a}) - \T_t\cdot(\ref{eq:GX-a})$, compared with (\ref{eq:TG}), means that

$$
4\T_t\T_3C_a + X_3X_t(\T_3G_t - \T_tG_3) =  4(\T_t\T_3 - X_3X_t(X_t - 4n))C_a,
$$
 
\noindent i.e.

\be
\T_tG_3 - \T_3G_t = 4\hat FC_a.    \label{eq:aTG}
\ee

\noindent Putting expression for $C_a$ from equation (\ref{eq:Ca}) of Appendix B in (\ref{eq:aTG}) proves (\ref{eq:a}). Multiplying (\ref{eq:GX+a}) and (\ref{eq:GX-a}) by $2\hat F$ and replacing $4\hat FC_a$ in them by the left-hand side of (\ref{eq:aTG}) yields

\be
2\hat F(4(X_+)_a^2 + X_3^2) = \T_t^2 - G_t^2 = \D_+X_t^2 - G_t^2,   \label{eq:+B}
\ee

\be
2\hat F(4X_a^2 + X_3^2) = \T_3^2 - G_3^2 = \D_-X_t^2 - G_3^2.   \label{eq:-B}
\ee

\noindent Adding (\ref{eq:+B}) and (\ref{eq:-B}) and using (\ref{eq:B+}) gives eq.~(\ref{eq:2+}), while using them to eliminate $(X_+)_a^2$ from eq.~(\ref{eq:A+}) and $X_a^2$ from (\ref{eq:A-}) leads, respectively, to eqs.~(\ref{eq:+}) and (\ref{eq:-}). 
\end{proof}

\par The system of equations obtained is still redundant since equations (\ref{eq:2+}), (\ref{eq:A2+}), (\ref{eq:+}) and (\ref{eq:-}) are in fact linearly dependent. Therefore The system of PDE in the previous theorem is equivalent to the one in the main corollary \ref{cor:main}.  












\section{Comparison with Toda lattice approach and new relations among different variables}   \label{sec:comp}

\par Let us compare with the results of Toda lattice approach in section \ref{sec:Toda}.  Recall that there we used functions


$$
F = UW = U_nW_n(1-c^2), \ \ \ \ \ G_+ = W\D_+ U - U\D_+ W, \ \ \ \ \ G_- =  W\D_- U - U\D_- W.
$$

\noindent Certain combination of equations (\ref{eq:Tt3}) of theorem \ref{theorem:4ord} from the previous section and (\ref{eq:x}) of the main corollary \ref{cor:main} is in fact exactly equivalent to the equation (\ref{eq:2}) from theorem \ref{thm:F4T} derived in section \ref{sec:Toda}. The combination is $2\hat F\cdot$(\ref{eq:Tt3})$ - $(\ref{eq:x}) as can be guessed by comparing its senior derivative terms of 4th and 3rd order with that of (\ref{eq:2}). It can be written as

$$
2\hat F\D_+\D_-X_t - \D_+X_t\D_-X_t + G_tG_3 - 2\hat F(\xi_+G_3 + \xi_-G_t) - 4X_3\hat F^2 = 0,
$$

\ni which clearly has the same form as (\ref{eq:2}). Comparing the terms in the two equations, we find simple correspondences:

$$
4F = \frac{\D_+r_t - \sigma^2\D_-r_3}{1 - \sigma^2} + 2n = -1/2(X_t - 4n) = -1/2\hat F,
$$

$$
G_+ = G_t/8,  \ \ \ \ G_- = G_3/8,
$$

$$
A_+ = 4((1-c^2)\prt_cr_t - \xi_+\sigma^2A^2), \ \ \ \ A_- = -4((1-c^2)\prt_cr_3 + \xi_-A^2),
$$

\noindent and so


$$
\sigma_3[r_a, \D_+ r_a] = -2c\prt_cr_t, \ \ \ \ \sigma_3[r_a, \D_- r_a] = 2c\prt_cr_3.   \eqno(\ref{eq:ccom})
$$

\ni These last two new important relations show that the main additional variables appearing in the TW approach to coupled (Gaussian) matrices -- the commutators on the left-hand side -- directly correspond to the main new objects appearing in the ASvM approach to the same problem -- the derivatives w.r.t.~the coupling $c$ (or w.r.t.~time $t$, recall $c = e^{-t}$) of the log-derivatives of the joint largest eigenvalue probability.
\par The Adler-van Moerbeke equation~\cite{AvM1}, formula (\ref{eq:AvM}) of section \ref{sec:Toda}, in the current variables reads:

$$
\D_+\frac{G_3}{\hat F} = \D_-\frac{G_t}{\hat F}    \eqno(\ref{eq:AvM})
$$


\section{The smallest complete set of PDE for two Gaussian coupled matrices}   \label{sec:fin4eq}
 
\par One can in fact eliminate the remaining auxiliary variables -- $A_+$ and $A_-$, which turned out to be directly related to the derivatives of $\ln\tau_n^J$ with respect to the coupling parameter $c$, and thus explicitly obtain PDE in terms of spectral endpoints only. To this end, let 

$$
F_t = \hat F + 2A^2, \ \ \ F_3 = \hat F + 2\sigma^2A^2, \ \ \ \Delta = F_tH_t^2 - F_3H_3^2, \ \ \ J = X_t^2 - X_3^2 - 4\sigma^2(A^2)^2,
$$

$$
J_+ = 2\sigma^2(2D_+^2 + A^2J), \ \ \ J_- = 2(2\sigma^2D_-^2 + A^2J), \ \ \ J_A = 4\sigma^2(D_+D_- - 2X_3(A^2)^2),
$$

$$
P_x = \T_t\T_3 - H_tH_3 - 2X_3X_t\hat F - J_A, \ \ \ P_t = \T_t^2 - H_t^2 - (2X_3^2+J)\hat F - J_+, \ \ \ P_3 = \T_3^2 - H_3^2 - (2X_3^2+J)\hat F - J_-,
$$

$$
P_+ = F_tP_t + F_3P_3, \ \ \ \ \ \ P_a = H_3^2P_t + H_t^2P_3 - 2H_tH_3P_x.
$$

\noindent Recall that 

$$
\hat F = X_t-4n,  \ \ \ \ \ \ D_+ = \D_+A^2,  \ \ \ \ \ D_- = \D_-A^2,
$$

\ni and introduce

$$
S_t = A^2\T_t - \hat FD_+, \ \ \ \ \ \ S_3 = A^2\T_3 - \hat FD_-, \ \ \ \ \ \ J_a = H_3\T_t - H_t\T_3.
$$

\bigskip
\par Then, expressing $A_+$, $A_-$ from the equations in the main corollary \ref{cor:main}, which are linear in them and their squares, and putting into the other ones, one gets four final equations:

\be
(H_tP_+ - 2F_3H_3P_x)(2F_tH_tP_x - H_3P_+) = 4\Delta^2J_A,  \label{eq:Axx}
\ee

\be
(H_tP_+ - 2F_3H_3P_x)^2 = 4\Delta(\Delta J_+ -2\sigma^2A^2P_a),  \label{eq:A+2}
\ee

\be
(2F_tH_tP_x - H_3P_+)^2 = 4\Delta(\Delta J_- + 2A^2P_a),  \label{eq:A-2}
\ee

\be
S_3(H_tP_+ - 2F_3H_3P_x) - S_t(2F_tH_tP_x - H_3P_+) =  2A^2\Delta J_a,  \label{eq:aa}
\ee

\noindent as one should since there are four independent senior derivatives involved here: 

$$
\T_t = \D_+X_t = \D_- X_3 - 2\D_+A^2, \ \ \T_3 = \D_-X_t = \D_+X_3 - 2\sigma^2\D_-A^2, \ \ D_+  = \D_+A^2, \ \ D_- = \D_-A^2.
$$

\noindent From the first three one derives:

\be
P_x^2 = \hat P_A + 2(H_t^2 - \sigma^2H_3^2)\frac{A^2P_a}{\Delta},  \label{eq:Px}
\ee

\be
P_+^2 = 4(F_tF_3\hat P_A + \Delta(F_tJ_+ - F_3J_-)) + 8(F_3^2H_3^2 - \sigma^2F_t^2H_t^2)\frac{A^2P_a}{\Delta},  \label{eq:P+}
\ee

\be
A^2P_a^2 = 2\Delta P_a(D_+^2 - \sigma^2D_-^2) + \Delta^2I_a,  \label{eq:P_a}
\ee

\noindent where

$$
\hat P_A = H_3^2J_+ + H_t^2J_- + 8\sigma^2H_tH_3(D_+D_- - 2X_3(A^2)^2),
$$

$$
I_a = J(2(D_+^2 + \sigma^2D_-^2) + A^2J) +16\sigma^2A^2X_3(D_+D_- - X_3(A^2)^2).
$$

\noindent We repeat the involved definitions once more for convenience:

$$
X_t = -2\D_+r_t - 2\sigma^2A^2 = -2\D_-r_3 - 2A^2, \ \ \ \ \ \ X_3 = -2\D_+r_3 = -2\D_-r_t, 
$$

$$
r_t = \D_+\ln\tau_n^J, \ \ \ \ \ \ r_3 = \D_-\ln\tau_n^J, \ \ \ \ \ \ A^2 = \frac{\D_+r_t - \D_-r_3}{1 - \sigma^2},
$$

$$
H_t = 4r_t - 2\xi_+\D_+r_t + \xi_-X_3, \ \ \ \ \ \ H_3 = 4r_3 - 2\xi_-\D_- r_3 + \xi_+X_3.
$$

\ni The combinations $P_x$ and $P_+$ turn into Painlev\'e IV equations in the one-matrix limit, while $P_a$, $P_A$, $I_a$ as well as $\Delta$ go to zero then. 
\par Since 

$$
\Delta = F_tH_t^2 - F_3H_3^2 \sim 2(1-\sigma^2)A^2(8G)^2,
$$

\noindent and, when $\xi_2 \to \infty$, $8G \equiv 4(r_t - \xi_1\prt_{\xi_1} r_t) = 4(r_3 - \xi_1\prt_{\xi_1} r_3)$, we can estimate

$$
P_a = H_3^2P_t + H_t^2P_3 - 2H_tH_3P_x \sim (8G)^2(C_1A^2 + C_2\prt_{\xi_1} A^2),
$$

\noindent where the quantities $C_1$ and $C_2$ remain finite in the next limit, so at fixed $c$ (or $\sigma^2$), as $\xi_2 \to \infty$,

$$
\frac{A^2P_a}{\Delta} \to 0.
$$

\ni Thus, equations (\ref{eq:Px}) and (\ref{eq:P+}) also tend to Painlev\'e IV in the one-matrix limit, while the other two -- (\ref{eq:P_a}) and (\ref{eq:aa}) -- become trivial.

\section{Conclusions}

The joint probability for the (largest) eigenvalues of two coupled Gaussian matrices with unitary invariant probability density satisfies a number of nonlinear integrable PDE some of which are coupled analogs of Painlev\'e IV equation for one-matrix GUE. The corresponding equations for the scaling limits of Airy process can be obtained from the system in the main corollary \ref{cor:main} of section \ref{sec:results}, but their consideration as well as due generalizations to several spectral endpoints (partly treated in section \ref{sec:Toda} though) and several coupled matrices are delegated to a forthcoming work.
\par The matrix kernel approach appears in a sense superior to the one based on Hirota bilinear identities and Virasoro constraints. The first allows to obtain at once {\it all} PDE satisfied by the joint gap probabilities while the last gives various subsets of the whole set of such PDE, depending on which nonlinear integrable equations in ``times" have been taken as the starting point. This situation is quite similar to what we recently found for the one-matrix case with several spectral gaps~\cite{Multipt1matr}.
\par For the Gaussian coupled ensemble simple relations among variables of different approaches are found here, and there is hope, supported by our analysis of single-matrix UE~\cite{UniUE}, that they again can be extended to other coupled RM. For some of them this is already shown here in section \ref{sec:unirel}. Comparison of {\it biorthogonal} function structures considered in~\cite{BeEyHa} and the matrix kernel approach of~\cite{TW2} is an interesting future direction. Biorthogonal case analogs of 3-term relations for functions in the resolvent kernel~\cite{UniUE}, which may combine finite difference recurrence and differentiation formulas, need to be found. They would present the coupled case from the (bi)orthogonal functions point of view, appearing the best for the  description of one-matrix ensembles. 
\par A Dyson Brownian motion (BM) model (or Dyson process) for all orthogonal-polynomial matrix ensembles can be constructed. It gives a class of coupled ensembles whose joint probability density satisfies a diffusion Fokker-Planck (FP) equation. The BM model can be mapped onto a quantum many-body problem and its transition density conveniently represented in terms of Green function for the quantum mechanical model, see e.g.~\cite{For05, Macedos}. Then, for unitary ensembles possessing the BM representation, one can get the matrix kernel of the form obtained in~\cite{TW2} for Hermite and Laguerre cases, i.e. express it in terms of the corresponding eigenfunctions and eigenvalues of the  effective Hamiltonian obtained from stationary FP operator by the above mapping. Then it is possible to give a TW-type derivation of PDE for arbitrary BM ensembles, if the differentiation formulas~\cite{TW1} are available, a situation just like in the one-matrix case. 
\par Current work also paves the way for obtaining new integrable PDE satisfied by gap probabilities of other ensembles with similar matrix kernel such as e.g.~Pfaffian ensembles. 

\bigskip
{\bf\large Acknowledgements} \\
Author is enormously grateful to C.A.Tracy for constant support and encouragement during several years as this work has been in progress. Useful discussions with A.Borodin, M.Bertola and S.-Y.Lee are also acknowledged. Author wishes to especially thank M.Adler for the critical discussion of the previous version of the paper and the referees for the suggestions which helped improve the text.
\par This work was done with partial NSF support under grant DMS-0906387 and VIGRE grant DMS-0636297.

\section*{Appendix A}

Proof of theorem \ref{thm:genT} is given here.
\par For further analysis we need more commutation relations:

$$
\A_0\frac{\tilde\A}{c} = \frac{\A}{c}(\A_0 + 1) - 2\A, \ \ \ \ \ \ \tilde\A_0\frac{\A}{c} = \frac{\A}{c}(\tilde\A_0 + 1) - 2\tilde\A.
$$

\ni The system of five equations we are considering apparently contains more unknowns than equations but we can enhance it. To this end, add up eq.~(\ref{eq:+t}) multiplied by $\A W$ and eq.~(\ref{eq:-t}) multiplied by $\A U$, which leads to the relation:

\be
2(\A W\A_0U - \A U\A_0W) = -2\A^2T\A F - \A\left(\frac{(\A F)^2 - G^2}{4F}\right).   \label{eq:hG}
\ee

\ni The ``dual" copy of it is the outcome of adding up eq.~(\ref{eq:+s}) multiplied by $\tilde\A W$ and eq.~(\ref{eq:-s}) multiplied by $\tilde\A U$:

\be
2(\tilde\A W\tilde\A_0U - \tilde\A U\tilde\A_0W) = -2\tilde\A^2T\tilde\A F - \tilde\A\left(\frac{(\tilde\A F)^2 - \tilde G^2}{4F}\right).   \label{eq:htG}
\ee

\noindent The expression on the left-hand side of (\ref{eq:hG}) can be written in two different ways:

$$
\A W\A_0U - \A U\A_0W = \A G_0 - (W\A\A_0U - U\A\A_0W) = -\A_0G + (W\A_0\A U - U\A_0\A W).
$$

\noindent Using the commutation relation

$$
\A_0\A = \A(\A_0 - 1),
$$

\ni one gets

$$
2(W\A\A_0U - U\A\A_0W) = (\A_0 + 1)G + \A G_0,
$$

\ni and so

$$
2(\A W\A_0U - \A U\A_0W) = \A G_0 - (\A_0 + 1)G.
$$

\ni The ``dual" of this also obviously holds:

$$
2(\tilde\A W\tilde\A_0U - \tilde\A U\tilde\A_0W) = \tilde\A\tilde G_0 - (\tilde\A_0 + 1)\tilde G.
$$

\ni Plugging the last expressions into the left-hand sides of eqs.~(\ref{eq:hG}) and (\ref{eq:htG}), respectively, one gets the needed additional equations 

\be
\A G_0 = (\A_0 + 1)G - 2\A^2T\A F - \A\left(\frac{(\A F)^2 - G^2}{4F}\right).   \label{eq:AG0}
\ee

\be
\tilde\A\tilde G_0 = (\tilde\A_0 + 1)\tilde G - 2\tilde\A^2T\tilde\A F - \tilde\A\left(\frac{(\tilde\A F)^2 - \tilde G^2}{4F}\right).   \label{eq:tAG0}
\ee

\noindent Using the found expressions for $G$, $\tilde G$ in terms of $T$ and commutation relations, we can express

$$
(\A_0 + 1)G = -(\A_0 + 1)(\A_0 - 1)\frac{\tilde\A T}{2c} = -(\A_0^2 - 1)\frac{\tilde\A T}{2c} =
$$

$$
= -1/2(\A_0(\tilde\A/c(\A_0 + 1) - 2\A) - \tilde\A/c)T = -1/2((\tilde\A/c(\A_0 + 1) - 2\A)(\A_0 + 1) - 2\A(\A_0 - 1) - \tilde\A/c)T =
$$

$$
= -1/2(\tilde\A/c(\A_0^2 + 2\A_0) - 4\A\A_0)T,
$$

\ni i.e.

\be
(\A_0 + 1)G = 2\A\A_0T - \frac{1}{2c}\tilde\A(\A_0^2 + 2\A_0)T.  \label{eq:A0G}
\ee

\ni Completely similarly,

\be
(\tilde\A_0 + 1)\tilde G = 2\tilde\A\tilde\A_0T - \frac{1}{2c}\A(\tilde\A_0^2 + 2\tilde\A_0)T.  \label{eq:tA0G}
\ee

\ni The last expressions together with (\ref{eq:0}) mean that the eqs. (\ref{eq:AG0}) and (\ref{eq:tAG0}) can be written as

$$
\A\left(G_0 + \frac{(\A F)^2 - G^2}{4F} - 2\A_0T\right) = -\tilde\A\left(\frac{1}{4c}(\A^2T)^2 + \frac{1}{2c}(\A_0^2 + 2\A_0)T\right),  \eqno(\ref{eq:AG0})
$$

$$
\tilde\A\left(\tilde G_0 + \frac{(\tilde\A F)^2 - \tilde G^2}{4F} - 2\tilde\A_0T\right) = -\A\left(\frac{1}{4c}(\tilde\A^2T)^2 + \frac{1}{2c}(\tilde\A_0^2 + 2\tilde\A_0)T\right). \eqno(\ref{eq:tAG0})
$$

\ni Now we eliminate $G_0$ and $\tilde G_0$ from them with the help of eqs. (\ref{eq:AF}) and (\ref{eq:tAF}) and obtain two higher-order, in fact 5-th order in $T$, PDE where everything can be expressed in terms of $T$ only, the ones in the theorem \ref{thm:genT}.

\section*{Appendix B}

\par Let us write out the diagonal and anti-diagonal parts of the matrix equations (\ref{eq:I1}), (\ref{eq:I2}) and (\ref{eq:I3}) separately. Equation (\ref{eq:I1}) gives, using (\ref{eq:X-d}),

$$
[(X_+)_a, G_a] = -2A^2\sigma\T_d + \{X_a, \T_a\},
$$

\noindent which splits into two scalar equations:

$$
[(X_+)_a, G_a] = -\T_tA^2\sigma, \ \ \  \{X_a, \T_a\} = \T_3A^2(\sigma_3\sigma).    
$$

\noindent Its anti-diagonal part reads:

\be
X_3G_a = G_3(X_+)_a + \T_t\sigma_3X_a.    \label{eq:I1a}
\ee

\noindent For the eq.~(\ref{eq:I2}) we get in the same way:

$$
[(X_+)_a, \T_a] = -2A^2\sigma G_d + \{X_a, G_a\},
$$

\noindent which splits into

$$
[(X_+)_a, \T_a] = -G_tA^2\sigma, \ \ \ \ \ \{X_a, G_a\} = G_3A^2(\sigma_3\sigma),    
$$

\noindent and the anti-diagonal part,

\be
X_3\T_a = \T_3(X_+)_a + G_t\sigma_3X_a.    \label{eq:I2a}
\ee

\noindent Combining the above equations implies some new ones, e.g.

\be
\{(X_+)_a, X_a\} = (\sigma_3\sigma)X_3A^2,    \label{eq:XXa}
\ee

\be
X_3[\T_a, G_a] = (G_tG_3 - \T_t\T_3)A^2\sigma,   \label{eq:TGd}
\ee

\be
X_3(\T_3G_a - G_3\T_a) = (\T_t\T_3 - G_tG_3)\sigma_3X_a.    \label{eq:TGa}
\ee

\noindent Taking the diagonal part of equation (\ref{eq:I3}) (using (\ref{eq:X-d}) again) gives two scalar equations, one of which is the above (\ref{eq:XXa}) again, while the other is

\be
[\T_a, G_a] = -\left(\frac{3X_t^2 + X_3^2}{2} - 8nX_t + 2(X_+)_a^2 + 2X_-^2\right)A^2\sigma.    \label{eq:I3d}
\ee

\noindent The anti-diagonal part of (\ref{eq:I3}) is

$$
\T_3\sigma_3G_a - G_3\sigma_3\T_a = (3\text{tr}X_d^2 - 8nX_t + 2X_-^2)X_a - X_3\sigma_3(X_3\sigma_3X_a + 2A^2\sigma(X_+)_a) + 2[(X_+)_a, X_a](X_+)_a,
$$

\noindent i.e.

\be
\T_3G_a - G_3\T_a = \left(\frac{3X_t^2 + X_3^2}{2} - 8nX_t + 6(X_+)_a^2 + 2X_-^2\right)\sigma_3X_a + 2(\sigma_3[(X_+)_a, X_a] - X_3A^2\sigma)(X_+)_a.   \label{eq:I3a}
\ee

\noindent One has an identity for $2\times2$ anti-diagonal matrices:

\be
2(X_+)_a^2X_a = (\{(X_+)_a, X_a\} - [(X_+)_a, X_a])(X_+)_a,   \label{eq:aId1}
\ee

\noindent which, together with (\ref{eq:XXa}) means

$$
2(X_+)_a^2\sigma_3X_a - (X_3A^2\sigma - \sigma_3[(X_+)_a, X_a])(X_+)_a = 0.   
$$

\noindent The last relation simplifies equation (\ref{eq:I3a}), it becomes

$$
\T_3G_a - G_3\T_a = \left(\frac{3X_t^2 + X_3^2}{2} - 8nX_t + 2(X_+)_a^2 + 2X_-^2\right)\sigma_3X_a.   \eqno(\ref{eq:I3a})
$$

\noindent Comparing this with eq.~(\ref{eq:TGa}) one gets a scalar relation

\be
\T_t\T_3 - G_tG_3 = X_3\left(\frac{3X_t^2 + X_3^2}{2} - 8nX_t + 2(X_+)_a^2 + 2X_-^2\right),   \label{eq:Ix}
\ee

\noindent which equally well follows from eqs.~(\ref{eq:TGd}) and (\ref{eq:I3d}). There is actually a lot of redundancy here and some relations arise in several different ways.
\par To summarize, we proved

\begin{lemma}

The full matrix first integrals reduce to four independent relations -- two scalar and two anti-diagonal matrix:

$$
\{(X_+)_a, X_a\} = (\sigma_3\sigma)X_3A^2,    \eqno(\ref{eq:XXa})
$$

$$
X_3G_a = G_3(X_+)_a + \T_t\sigma_3X_a,    \eqno(\ref{eq:I1a})
$$

$$
X_3\T_a = \T_3(X_+)_a + G_t\sigma_3X_a,    \eqno(\ref{eq:I2a})
$$






$$
\T_t\T_3 - G_tG_3 = X_3\left(\frac{3X_t^2 + X_3^2}{2} - 8nX_t + 2(X_+)_a^2 + 2X_-^2\right),   \eqno(\ref{eq:Ix})
$$

\noindent and all the other above relations follow by taking commutators or anti-commutators involving (\ref{eq:I1a}) and/or (\ref{eq:I2a}).

\end{lemma}

\par Next we use the identity (of which eq. (\ref{eq:aId1}) is an instance)

\begin{lemma}
For any anti-diagonal $2\times2$ matrices $A_a$ and $B_a$:

$$
2A_a^2\cdot B_a = (\{A_a, B_a\} - [A_a, B_a])A_a
$$

\noindent  and its scalar consequence

$$
4A_a^2B_a^2 = \{A_a, B_a\}^2 - [A_a, B_a]^2.
$$

\end{lemma}

\noindent In particular, matrices $(X_+)_a$ and $X_a \equiv (X_-)_a$ can be expressed as

$$
2A^2(X_{\pm})_a = (\{A, X_{\pm}\} - [A, X_{\pm}])A.
$$

\noindent Besides, we have the facts:

\be
[A, (X_+)_a] = \D_+A^2\sigma_3, \ \ [A, X_a] = -\D_- A^2\sigma, \ \ \{(X_+)_a, X_a\} = (\sigma_3\sigma)X_3A^2,
\ee

\noindent first two are consequences of the eqs.~(\ref{eq:ra}) and the third is the eq.~(\ref{eq:XXa}). We apply all this to derive

\begin{lemma}

\be
16\sigma^2A^2(X_+)_a^2 = A_+^2 - 4\sigma^2(\D_+A^2)^2,      \label{eq:A+}
\ee

\be
16A^2X_a^2 = A_-^2 - 4\sigma^2(\D_- A^2)^2,       \label{eq:A-}
\ee

\be
4A^2C_a \equiv 4A^2\sigma_3[(X_+)_a, X_a] = -A_+\D_-A^2 + A_-\D_+A^2,        \label{eq:Ca}
\ee

$$
A_+A_- = 4\sigma^2(\D_+A^2\D_-A^2 - 2X_3(A^2)^2).        \eqno(\ref{eq:Ax})
$$

\noindent The last equations should be considered as four more scalar first integrals. 

\end{lemma}

\par Let us substitute the expressions (\ref{eq:I1a}) and (\ref{eq:I2a}) into the right-hand sides of the equations of the system. 
The anti-diagonal parts of the equations, besides (\ref{eq:ra}), become:

\be
X_3\D_+(X_+)_a = \T_3(X_+)_a + (G_t - \xi_-X_3)\sigma_3X_a - 2X_3A^2\sigma\tilde A,    \label{eq:+X+a}
\ee

\be
\sigma_3\D_+X_a = X_3\tilde A - \xi_-(X_+)_a,   \label{eq:+X-a}
\ee

\be
X_3\D_+\T_a = (\xi_+G_3 + X_3(3X_t - 8n))(X_+)_a + \xi_+\T_t\sigma_3X_a + X_3G_t\tilde A,   \label{eq:+Ta}
\ee

\be
X_3\D_+G_a = \xi_+\T_3(X_+)_a + 2X_3A^2\sigma(X_+)_a + (\xi_+G_t + X_3^2)\sigma_3X_a + X_3\T_t\tilde A,   \label{eq:+Ga}
\ee

\be
\D_-(X_+)_a = -X_3\sigma_3A - \xi_+\sigma_3X_a,   \label{eq:-X+a}
\ee

\be
X_3\sigma_3\D_-X_a = (G_3 - \xi_+X_3)(X_+)_a + \T_t\sigma_3X_a + 2X_3A^2\tilde A,    \label{eq:-X-a}
\ee

\be
X_3\D_-\T_a = (\xi_-G_3 + X_3^2)(X_+)_a + (\xi_-\T_t + 2X_3A^2\sigma)\sigma_3X_a - X_3\T_3\sigma_3A.   \label{eq:-Ta}
\ee

\be
X_3\D_-G_a = \xi_-\T_3(X_+)_a + (\xi_-G_t + X_3(3X_t - 8n))\sigma_3X_a - X_3G_3\sigma_3A,   \label{eq:-Ga}
\ee

\noindent Applying the above anti-diagonal identities simplifies eqs.~(\ref{eq:+X+a}) and (\ref{eq:-X-a}), they turn into

$$
X_3\D_+(X_+)_a = \D_+X_3(X_+)_a + (G_t - \xi_-X_3 - A_+)\sigma_3X_a \equiv \D_+X_3(X_+)_a + R_t\sigma_3X_a,    \eqno(\ref{eq:+X+a})
$$

$$
X_3\sigma_3\D_-(X_-)_a = (G_3 - \xi_+X_3 - A_-)(X_+)_a + \D_-X_3\sigma_3X_a \equiv R_3(X_+)_a + \D_+X_3\sigma_3X_a,    \eqno(\ref{eq:-X-a})
$$

\noindent where

\be
R_t = 4r_t - 2\xi_+\D_+r_t, \ \ \ \ \ \ R_3 = 4r_3 - 2\xi_-\D_-r_3.    \label{eq:R}
\ee

\noindent As follows from these definitions,

\be
\D_-R_t = \xi_+\D_+X_3 - 2X_3, \ \ \ \ \ \ \D_+R_3 = \xi_-\D_-X_3 - 2X_3.   \label{eq:DR}
\ee

\noindent We now match the mixed second derivatives $\D_+\D_-(X_+)_a$ etc. found from the corresponding pairs of the above anti-diagonal equations. It turns out that taking both pairs of equations (\ref{eq:+X+a}), (\ref{eq:-X+a}) and (\ref{eq:+X-a}), (\ref{eq:-X-a}), as well as plugging (\ref{eq:I1a}) and (\ref{eq:I2a}) into the left-hand sides of (\ref{eq:+Ga}) and (\ref{eq:-Ta}), respectively, leads to the same higher-order scalar consistency equation:

\begin{lemma}

$$
X_3\D_+\D_-X_3 - \D_+X_3\D_-X_3 + R_tR_3 - (X_3 + \xi_+\xi_-)X_3^2 = 0.   \eqno(\ref{eq:S})
$$

\noindent This is in fact a 4-th order PDE for $\ln\tau_n^J$, since all the variables in (\ref{eq:S}) have been already expressed in terms of $\ln\tau_n^J$ before. 

\end{lemma}

\noindent However, we are going to obtain 3-rd order equations in $\ln\tau_n^J$ and show that (\ref{eq:S}) is a consequence of them. 
\par Consider next the eqs.~(\ref{eq:+X-a}) and (\ref{eq:-X+a}). Applying (\ref{eq:ra}), one sees that both of them then reduce to the same anti-diagonal consistency equation:

\be
\D_+\D_-A = (X_3 + \xi_+\xi_-)A.   \label{eq:DDAa}
\ee

\ni Using it, compute

$$
\D_+\D_-A^2 = \{A, \D_+\D_-A\} + \{\D_+A, \D_-A\} = 2(X_3 + \xi_+\xi_-)A^2 + \{\D_+A, \D_-A\},
$$

\ni and, applying (\ref{eq:ra}) again and (\ref{eq:XXa}),

$$
\sigma^2\{\D_+A, \D_-A\} = -1/2\xi_+\sigma^2A_- - 1/2\xi_-A_+ + (X_3 - 2\xi_+\xi_-)\sigma^2A^2.
$$

\ni Equation (\ref{eq:DDA}) of theorem \ref{theorem:4ord} immediately follows from the last two formulas.  
 
\begin{lemma}

The remaining anti-diagonal equations (\ref{eq:+Ta}) and (\ref{eq:-Ga}), if we substitute (\ref{eq:I2a}) and (\ref{eq:I1a}), respectively, into their left-hand sides, result in three independent scalar equations (three rather than four because $\D_+\T_3 = \D_-\T_t = \D_+\D_-X_t$):

$$
\D_+\T_3 = \D_-\T_t = \xi_-G_t + \xi_+G_3 + X_3(3X_t - 8n),   \eqno(\ref{eq:Tt3})
$$

\be
X_3\D_+G_t = \D_+X_3G_t - R_t\T_3 + \xi_+X_3\T_t,    \label{eq:+Gt}
\ee

\be
X_3\D_-G_3 = \D_-X_3G_3 - R_3\T_t + \xi_-X_3\T_3.    \label{eq:-G3}
\ee

\end{lemma}

\noindent The last two of them, however, have been in fact already integrated when $G_d$ was obtained: they are trivially satisfied if $G_d$ is plugged into their left-hand sides.

\par The other 4-th order equations can be produced in more than one way. We write out just one here.

\begin{lemma}

The following higher-order (4-th order in $\ln\tau_n^J$) equations hold:

$$
X_3\D_+\T_t = \D_+X_3\T_t - R_tG_3 + \xi_+X_3G_t + X_3(X_3^2 + 4(X_+)_a^2),    \eqno(\ref{eq:+Tt})
$$

$$
X_3\D_-\T_3 = \D_-X_3\T_3 - R_3G_t + \xi_-X_3G_3 + X_3(X_3^2 + 4X_a^2).    \eqno(\ref{eq:-T3})
$$

\end{lemma}

\begin{proof}

Consider the remaining diagonal equations of the original system, namely the ones for the derivatives of $\T_d$:

\be
\D_+\T_d = 3(X_d^2 + (X_+)_a^2) - 8nX_d + X_-^2 + 2\xi G_d + \{\tilde A, G_a\},   \label{eq:+Td}
\ee

\be
\D_-\T_d = 2\sigma_3(X_d^2 + (X_+)_a^2 - 4nX_d + X_-^2) + \sigma_3(X_d^2 - (X_+)_a^2) - \sigma_3((X_-)_d^2 - X_a^2) + 2\sigma_3\xi G_d + [A, \T_a].   \label{eq:-Td}
\ee

\noindent Substituting (\ref{eq:I1a}) into (\ref{eq:+Td}) and (\ref{eq:I2a}) into (\ref{eq:-Td}), then separating the Tr$(\cdot)$ and Tr$(\sigma_3\cdot)$ parts, we get equation (\ref{eq:Tt3}) again plus two more scalar equations:

$$
X_3\D_+\T_t = X_3\left(\frac{3X_t^2 + X_3^2}{2} - 8nX_t + 2(X_+)_a^2 + 2X_-^2\right) + X_3(X_3^2 + 4(X_+)_a^2) + 
$$

\be
+ \xi_+X_3G_t + (\xi_-X_3 + A_+)G_3 + 2\sigma^2\D_-A^2\T_t,   \label{eq:+Ttd}
\ee

$$
X_3\D_-\T_3 = X_3\left(\frac{3X_t^2 + X_3^2}{2} - 8nX_t + 2(X_+)_a^2 + 2X_-^2\right) + X_3(X_3^2 + 4X_a^2) + 
$$

\be
+ (\xi_+X_3 + A_-)G_t + \xi_-X_3G_3 + 2\D_+A^2\T_3.   \label{eq:-T3d}
\ee

\noindent Now we recall relations (\ref{eq:TtX}), (\ref{eq:T3X}) and (\ref{eq:Ix}) and see that these last equations reduce to equations (\ref{eq:+Tt}) and (\ref{eq:-T3}), respectively. 

\end{proof}

This, together with formulas (\ref{eq:+B}) and (\ref{eq:-B}) applied in eqs.~(\ref{eq:+Tt}) and (\ref{eq:-T3}), ends the proof of theorem \ref{theorem:4ord}.

\section*{Appendix C}

In Appendix B we proved theorem \ref{theorem:4ord} -- the five 4th-order equations for $\ln\tau_n^J$.







\ni After using formulas 

$$
2\hat F(4(X_+)_a^2 + X_3^2) = \T_t^2 - G_t^2,   \eqno(\ref{eq:+B})
$$

$$
2\hat F(4X_a^2 + X_3^2) = \T_3^2 - G_3^2,   \eqno(\ref{eq:-B})
$$

\ni and $\T_t = \D_+X_t$, $\T_3 = \D_-X_t$, eqs.~(\ref{eq:+Tt}) and (\ref{eq:-T3}) take the final form:

$$
2\hat FX_3\D_+^2X_t = 2\hat F(\D_+X_3\D_+X_t - R_tG_3 + \xi_+X_3G_t) + X_3(\D_+X_t^2 - G_t^2),    \eqno(\ref{eq:+Tt})
$$

$$
2\hat FX_3\D_-^2X_t = 2\hat F(\D_-X_3\D_-X_t - R_3G_t + \xi_-X_3G_3) + X_3(\D_-X_t^2 - G_3^2).    \eqno(\ref{eq:-T3})
$$
 
On the other hand, in the corollary \ref{cor:main} of section \ref{sec:results} we obtained the system of 3rd-order equations, which should be an integrated version of equations of theorem \ref{theorem:4ord} and other higher-order equations. Here we show that indeed all equations of theorem \ref{theorem:4ord} follow from the ones of corollary \ref{cor:main}.









\begin{proof}

One has the following equations for the derivatives of $G_t, G_3, A_+$ and $A_-$:
  
$$
X_3\D_+G_t = \D_+X_3G_t - R_t\D_-X_t + \xi_+X_3\D_+X_t,    \eqno(\ref{eq:+Gt})
$$

$$
X_3\D_-G_3 = \D_-X_3G_3 - R_3\D_+X_t + \xi_-X_3\D_-X_t.    \eqno(\ref{eq:-G3})
$$

\be
2\hat F\D_-G_t = -(G_3 - 2\xi_-\hat F)\D_+X_t + (G_t + 2\xi_+\hat F)\D_-X_t,    \label{eq:-Gt}
\ee

\be
2\hat F\D_+G_3 = (G_3 + 2\xi_-\hat F)\D_+X_t - (G_t - 2\xi_+\hat F)\D_-X_t,    \label{eq:+G3}
\ee

\be
X_3\D_+A_+ = \D_+X_3A_+ + 2\sigma^2(R_t\D_-A^2 - \xi_+X_3\D_+A^2),    \label{eq:+At}
\ee

\be
X_3\D_-A_- = \D_-X_3A_- + 2(R_3\D_+A^2 - \xi_-X_3\D_-A^2).    \label{eq:-A3}
\ee

\be
2A^2\D_-A_+ = -(A_- + 4\xi_-A^2)\D_+A^2 + (A_+ - 4\xi_+\sigma^2A^2)\D_-A^2,    \label{eq:-At}
\ee

\be
2A^2\D_+A_- = (A_- - 4\xi_-A^2)\D_+A^2 - (A_+ + 4\xi_+\sigma^2A^2)\D_-A^2.    \label{eq:+A3}
\ee

\ni (Writing them out, we used e.g.~formulas $4\hat FC_a = G_3\D_+X_t - G_t\D_-X_t$ and $4A^2C_a = A_-\D_+A^2 - A_+\D_-A^2$.) Recall also relations:

$$
\D_+^2X_3 = \D_+\D_-X_t + 2\sigma^2\D_+\D_-A^2, \ \ \ \ \ \ \D_-^2X_3 = \D_+\D_-X_t + 2\D_+\D_-A^2,
$$

$$
\D_+\D_-X_3 = \D_-^2X_t + 2\sigma^2\D_-^2A^2 = \D_+^2X_t + 2\D_+^2A^2,
$$

\ni which are consequences of 

$$
\D_+X_3 = \D_-X_t + 2\sigma^2\D_-A^2, \ \ \ \ \ \D_-X_3 = \D_+X_t + 2\D_+A^2,
$$

\ni and which show that the written above are in fact the only {\it five} independent 4th-order equations. The rest of the proof is long, tedious calculations, in which we take derivatives of the 3rd-order PDE and compare them with such combinations of the 4th-order PDE that the highest derivative terms cancel out. Then we see that the rest reduces to the 3rd-order equations themselves. In this way one verifies e.g.~that

$$
\D_+(\ref{eq:x}) = \D_-X_t\cdot(\ref{eq:+Tt}) + 2\hat FX_3\D_+X_t\cdot(\ref{eq:Tt3}),
$$

\ni and, similarly, or by symmetry,

$$
\D_-(\ref{eq:x}) = \D_+X_t\cdot(\ref{eq:-T3}) + 2\hat FX_3\D_-X_t\cdot(\ref{eq:Tt3}).
$$

\ni Differentiating (\ref{eq:Ax}), one can get a combination

$$
\hat FX_3\D_+(\ref{eq:Ax}) + \sigma^2\D_-A^2\cdot(\ref{eq:+Tt}) = 0,
$$

\ni which turns out to be equivalent to 

$$
\sigma^2\cdot(\ref{eq:S}) + X_3\D_+A^2\cdot(\ref{eq:DDA}) = 0.
$$

\ni Then by symmetry one has also

$$
\sigma^2\cdot(\ref{eq:S}) + X_3\D_-A^2\cdot(\ref{eq:DDA}) = 0
$$

\ni from the combination

$$
\hat FX_3\D_-(\ref{eq:Ax}) + \sigma^2\D_+A^2\cdot(\ref{eq:-T3}) = 0.
$$

\ni If one considers equation $\hat FX_3\D_+(\Pe_t + \Pe_3) = 0$ and applies (\ref{eq:+-}) and (\ref{eq:a}), one comes to the combination

$$
\D_+X_t\cdot(\ref{eq:+Tt}) + 2\hat FX_3\D_-X_t\cdot(\ref{eq:Tt3}) = 0,
$$

\ni and analogously or by symmetry one also gets the combination

$$
\D_-X_t\cdot(\ref{eq:-T3}) + 2\hat FX_3\D_+X_t\cdot(\ref{eq:Tt3}) = 0.
$$

Thus, one sees that all five independent 4th-order equations indeed follow from the 3rd-order equations. 
\end{proof}



\end{document}